\documentclass{llncs}

\usepackage{amssymb}
\usepackage{amsmath}
\usepackage{stmaryrd}
\usepackage[hidelinks,breaklinks]{hyperref}
\usepackage{proof-dashed}
\usepackage{multirow}
\usepackage{tikz}
\usetikzlibrary{arrows}
\usetikzlibrary{positioning}
\usepackage{listings}
\usepackage{calrsfs}
\usepackage{todonotes}
\usepackage{breakurl}

\usepackage{enumitem}

\newcommand{\qedhere}{\qed}

\newcommand{\m}[1]{\ensuremath{\mathsf{#1}}}

\newcommand{\pid}[1]{\m{#1}}
\newcommand{\pids}[1]{\til {\pid #1}}

\newcommand{\rhop}{\rho_{\pid p}}
\newcommand{\rhoq}{\rho_{\pid q}}

\newcommand{\gentell}{\tell{\pid p}{\pid q}{\pid r}}
\newcommand{\tell}[3]{{#1}\!: {#2}\, \code{<->}\, {#3}}
\newcommand{\genatell}{\pid p\!: \bullet_{\pid q}\stackrel{x,y}{\code{<->}}\bullet_{\pid r}}
\newcommand{\genatold}{\acom{\bullet_{\pid p}.\pid r}{\pid q}{\maybe {\pid p}}}
\newcommand{\gentellsubscript}{\pid p: \pid q\, \code{<->}\, \pid r}
\newcommand{\emptyN}{\nil} 

\newcommand{\code}[1]{\mbox{\upshape{\texttt{#1}}}}
\newcommand{\nil}{\boldsymbol 0}

\newcommand{\com}[2]{#1\;\code{-\hspace{-0.3mm}>}\;#2}
\newcommand{\acom}[3]{#1\stackrel{#3}{\code{-\hspace{-0.3mm}>}}#2}

\newcommand{\gencomf}{\com{\pid p.e}{\pid q.f}}
\newcommand{\assend}[4]{\acom{{#1}.{#2}}{\bullet_{#3}}{#4}}
\newcommand{\asrecv}[4]{\acom{\bullet_{#1}}{{#2}.{#3}}{#4}}
\newcommand{\genasend}{\assend{\pid p}{e}{\pid q}{x}}
\newcommand{\genarecv}{\asrecv{\pid p}{\pid q}{f}{\maybe v}}
\newcommand{\maybe}[1]{\hat{#1}}
\newcommand{\sel}[3]{\com{#1}{#2 [#3]}}
\newcommand{\gensel}{\sel{\pid p}{\pid q}{l}}
\newcommand{\genasels}{\acom{\pid p}{\bullet_{\pid q}[l]}x}
\newcommand{\genaselr}{\acom{\bullet_{\pid p}}{\pid q[l]}{\maybe l}}

\newcommand{\lleft}{\textsc{l}}
\newcommand{\lright}{\textsc{r}}
\newcommand{\equivd}{\equiv_{\defs}}
\newcommand{\defs}{\mathcal{D}}
\newcommand{\pdefs}{\mathcal{B}}
\newcommand{\tod}{\to_{\defs}}
\newcommand{\tob}{\to_{\pdefs}}
\renewcommand{\merge}{\sqcup}
\newcommand{\cond}[3]{\m{if}\, #1 \, \m{then} \, #2 \, \m{else} \, #3}
\newcommand{\eqcom}[2]{#1\;\code{<=}\;#2}

\newcommand{\gencond}{\cond{\eqcom{\pid p}{\pid q}}{C_1}{C_2}}

\newcommand{\gencondE}{\cond{\pid p.e}{C_1}{C_2}}
\newcommand{\start}[2]{#1 \, \m{start} \, #2}
\newcommand{\astart}[2]{\m{start} \, #1 \triangleright #2}
\newcommand{\atells}[2]{{#1}!!{#2}}
\newcommand{\atellr}[2]{\arecv{#1}{#2}}

\newcommand{\genstart}{\start{\pid p}{\pid q^T}}

\newcommand{\gendef}{X(\wtil{\pid q^T}) = C}
\newcommand{\genpdef}{X(\pids q) = B}

\newcommand{\pn}{\m{pn}}

\newcommand{\callP}[2]{#1\langle #2 \rangle}

\newcommand{\gencallP}{\callP{X}{\pids p}}
\definecolor{light-gray}{gray}{0.928}

\newcommand{\pcont}{\ast}
\newcommand{\til}{\tilde}
\newcommand{\wtil}{\widetilde}
\newcommand{\amend}{\m{Amend}}

\newcommand{\epp}[2]{[\![#1]\!]_{#2}}

\newcommand{\eval}[2]{{#1}\downarrow{#2}}

\newcommand{\rname}[2]{\ensuremath{\left\lfloor\mbox{{#1}$|${#2}}\right\rceil}}

\newcommand{\asend}[2]{{#1}!{#2}}
\newcommand{\arecv}[2]{#1?#2}
\newcommand{\actor}[4][]{{#2} \triangleright^{#1}_{#3} {#4}}
\newcommand{\parp}{\, \boldsymbol{|} \, }
\newcommand{\asel}[2]{{#1}\oplus#2}
\newcommand{\abranch}[2]{{#1}\&{#2}}
\newcommand{\precongr}{\preceq}
\newcommand{\precongrd}{\preceq_{\defs}}
\newcommand{\precongrb}{\preceq_{\pdefs}}
\newcommand{\defeq}{\stackrel{\Delta}{=}}

\newcommand{\acspar}{\hspace{0.8mm} | \hspace{0.8mm} }
\newcommand{\seq}{\vdash}
\newcommand{\cseq}{\triangleright}

\newcommand{\smallpar}[1]{\smallskip\noindent \textbf{\textit{#1.}}}

\newcommand{\tjudge}[4]{{#1};{#2} \seq {#3} \cseq {#4}}

\newcommand{\ACrich}{PC}

\newcommand{\knows}[3][G]{\ensuremath{{#2}\stackrel{#1}{\longleftrightarrow}{#3}}}
\newcommand{\update}[3][G]{\ensuremath{{#1}\cup\{{#2}\leftrightarrow{#3}\}}}
\newcommand{\know}[3][G]{\ensuremath{{#2}\stackrel{#1}{\rightarrow}{#3}}}
\newcommand{\updates}[3][G]{\ensuremath{{#1}\cup\{{#2}\rightarrow{#3}\}}}

\newcommand{\types}{\mathbb T}

\newcommand{\fst}{\m{fst}}
\newcommand{\snd}{\m{snd}}

\newcommand{\lcfa}[3][\til{G_i}]{\ensuremath{\lcfaname^{#1}_{#2}(#3)}}
\newcommand{\lcfb}[4][\til{G_i}]{\ensuremath{\lcfbname^{#1}_{#2}(#3,#4)}}
\newcommand{\lcfaname}{\m{fwd}}
\newcommand{\lcfbname}{\m{bck}}
\newcommand{\something}[2]{(\!|{#1}|\!)_{#2}}

\begin{document}
\pagestyle{plain}

\setlist[description]{leftmargin=2ex}



\title{
%
A Language for the Declarative Composition of Concurrent Protocols
%
}

\author{Lu\'is Cruz-Filipe and Fabrizio Montesi}
\institute{University of Southern Denmark \qquad \email{\{lcf,fmontesi\}@imada.sdu.dk}}

\maketitle

\begin{abstract}
%
A recent study of bugs in real-world concurrent and distributed systems found
that, while implementations of individual  protocols tend to be robust, the
composition of multiple protocols and its interplay with internal computation is
the culprit for most errors.

Multiparty Session Types and Choreographic Programming are methodologies for
developing correct-by-construction concurrent and distributed software, based on
global descriptions of communication flows. However, protocol composition is
either limited or left unchecked. Inspired by these two methodologies, in this
work we present a new language model for the safe composition of protocols,
called Procedural Choreographies (PC).

Protocols in PC are procedures, parameterised on the processes that enact them.
Procedures define communications declaratively using global descriptions, and
programs are written by invoking and composing these procedures. An
implementation in terms of a process model is then mechanically synthesised,
guaranteeing correctness and deadlock-freedom. We study PC in the settings of
synchronous and asynchronous communications, and illustrate its expressivity
with some representative examples.

\end{abstract}




\section{Introduction}\label{sec:intro}
%
In the last decades,
advances in multi-core hardware and large-scale networks have made concurrent and distributed systems widespread.
Unfortunately, programming such systems is a notoriously error-prone activity.
In \cite{LPSZ08}, $105$ randomly-selected concurrency bugs from real-world software projects are analysed.
Of these, $31$ are caused by deadlocks. Of the $74$ remaining ones, $97\%$ are caused by
violations of the programmer's intentions on atomicity or the ordering of actions.

The theory of Multiparty Session Types tackles these problems in communication protocols~\cite{HYC16}.
The idea is that developers express their intentions on communications declaratively, by writing protocol
specifications from a global viewpoint using an ``Alice and Bob'' notation.
Given such a global specification, an EndPoint Projection (EPP) mechanically synthesises the local specifications
of the I/O actions for each participant.
The local specifications are thus correct by construction.
Then, a type system can be used to check that implementations in process models follow
the generated local specifications. Multiparty session types guarantee that the implementation of each protocol, taken
in isolation, is deadlock-free and faithful to its global specification -- which, by its declarative nature, respects the
programmer's intentions. However, errors can still occur due to the composition of multiple protocol executions, which
is left unchecked. Different approaches for avoiding deadlocks in protocol compositions have been proposed
\cite{CMSY15,CDYP16}, but these limit how protocols can be composed -- e.g., protocols can be instantiated only in
a certain order, or connections among processes should form a tree structure -- and do not offer a means to easily and 
correctly translate the programmer's intentions on ordering: the language for composing protocols is not 
declarative.

A more recent empirical study~\cite{LLLG16} reveals that protocol
composition deserves more attention. It presents a taxonomy of 104 Distributed
Concurrency (DC) bugs. The authors' findings include a significant insight, which
we quote (emphasis ours):

\begin{quote}
``Real-world DC bugs are hard to find because many of them
linger in complex concurrent executions of multiple protocols.
[\ldots]
\emph{Individual protocols tend to be robust in general.} Only 18
DC bugs occur in individual protocols without any input
fault condition [\ldots].
\emph{On the other hand, a large majority of DC bugs happen
due to concurrent executions of multiple protocols} [\ldots]''
\end{quote}

In this quote, a protocol is not intended as an abstract specification as
in multiparty session types, but rather as the concrete series of events that
happen in an implementation, including internal computation.
This motivates us to address the problem of protocol composition in the
framework of Choreographic Programming~\cite{M13:phd}: a paradigm similar to
Multiparty Session Types, where global descriptions of communications are not
used as types, but rather as programs. In such programs, called choreographies,
developers declaratively program communications among processes together with
the internal computations that they perform. EPP is then used to synthesise an
implementation in terms of a process model~\cite{CHY12}. However, state of the
art models for choreographic programming still do not support many of the features used in real-world programs, such as 
those in~\cite{LLLG16}. One prominent aspect, which is the focus of this paper, is that these models do not support 
arbitrary compositions of protocol executions, since they are not modular. In particular, the absence
of full procedural abstraction disallows the development of reusable libraries
that can be composed as ``black boxes''.

Inspired by these observations, we propose a language model for
the correct programming of concurrent and distributed systems based on message
passing, called Procedural Choreographies (PC). PC is a new model for
choreographic programming, where we can define a protocol execution as a
procedure, parameterised on the processes that will actually enact it.
Composition of protocol executions is then obtained by allowing for the
arbitrary composition of procedure calls, a feature lacked by previous
choreography models. Nevertheless, PC inherits all the good properties of
languages based on declarative global descriptions of communications -- as in
choreographic programming and multiparty session types -- and extends them to
protocol compositions. In particular, the process implementations synthesised from choreographies are correct by 
construction and deadlock-free. Thus, PC contributes to bringing the current body of work on safe concurrent 
programming nearer to dealing with the kinds of bugs analysed in~\cite{LLLG16}.

\begin{example}
\label{ex:mergesort}
We discuss a parallel version of merge sort in PC. Although this is a toy example,
it cannot be written in any previous model for choreographic programming, thus
allowing us to present the key features of our work for a simple scenario.
More realistic and involved examples are presented after the formal presentation of PC.
  We make the standard assumption that we have concurrent processes with local storage and computational capabilities.
  In this example, each process stores a list and can use the following local functions:
  \lstinline$split1$ and \lstinline$split2$, respectively returning the first or second half of a list;
  \lstinline$is_small$, which tests if a list has at most one element;
  and \lstinline$merge$, which combines two sorted lists into one.
  The following (choreographic) procedure, \lstinline$MS$, implements merge sort on the list stored at its parameter process
  \lstinline$p$.\footnote{In the remainder, we use a \lstinline+monospaced+ font for readability of our
concrete examples, and other fonts for distinguishing syntactic categories in our formal arguments as
usual.}
\begin{lstlisting}
MS(p) =	if p.is_small then 0
	else p start q1,q2; p.split1 -> q1; p.split2 -> q2;
             MS<q1>; MS<q2>; q1.* -> p; q2.* -> p.merge
\end{lstlisting}
Procedure \lstinline$MS$ starts by checking whether the list at process \lstinline$p$ is small,
in which case it does not need to be
sorted (\lstinline$0$ denotes termination); otherwise, \lstinline$p$ starts two other processes \lstinline$q1$ and \lstinline$q2$
(\lstinline$p start q1,q2$), to which it
respectively sends the first and the second half of the list (\lstinline$p.split1 -> q1$ and \lstinline$p.split_2 -> q2$).
The procedure is recursively reapplied to \lstinline$q1$ and \lstinline$q2$, which independently (concurrently)
proceed to ordering their respective sub-lists. When this is done, \lstinline$MS$ stores the first ordered half from
\lstinline$q1$ to \lstinline$p$ (\lstinline$q1.* -> p$, where \lstinline$*$ retrieves the data stored in
\lstinline$q1$) and
merges it with the ordered sub-list from \lstinline$q2$ (\lstinline$q2.* -> p.merge$).
\end{example}

Our merge sort example showcases the key desiderata for PC:
\begin{description}
\item[General recursion.] Procedure calls can be followed by arbitrary code.
\item[Parameterised procedures.]
Procedures
are parametric on their processes
(\lstinline$p$ in
  \lstinline+MS+),
and can thus be reused 
with different processes (as in \lstinline+MS<q1>+ and \lstinline+MS<q2>+).
\item[Process spawning.] The ability of starting new processes.
There must be no bound on how many
  processes can be started, since this is decided at runtime. (In \lstinline+MS+, the number of spawned processes
  depends on the size of the initial list.)
\item[Implicit Parallelism.]
In the composition \lstinline+MS<q1>; MS<q2>+, the two calls can be run in parallel because they involve separate 
processes and are thus non-interfering. Thus, the code synthesised from this program should be parallel.
\end{description}

In the remainder, we explore more sophisticated programs
that require additional features, e.g., mobility of process names.

\subsection{Contributions}
We summarise our development and contributions.

%

\smallpar{Procedural Choreographies}
We introduce Procedural Choreographies (PC), a new
language model that supports all the features discussed
above (\S~\ref{sec:pc}).
%
%
We evaluate the expressivity of PC not only with our concurrent merge sort example, but also with a more involved
parallel downloader. This example makes use of additional features:
mobility of process names (networks with connections that evolve at runtime) and propagation of choices among
processes. It also makes heavy use of implicit parallelism to deal with the parallelisation of multiple streams.

\smallpar{Asynchrony}
PC can be endowed with both a synchronous (\S~\ref{sec:pc}) and an asynchronous semantics (\S~\ref{sec:async}). All
our results hold for both versions, and the two semantics enjoy a strong correspondence result
(Theorem~\ref{thm:pc-vs-apc}). This allows developers to reason about their programs in the simpler synchronous
setting and then to use the asynchronous semantics to obtain a more concurrent implementation, without worrying
about introducing unsafe behaviour.

\smallpar{Typing}
Mobility of process names requires careful handling, especially its interplay with procedure composition: if 
a procedure specifies that two processes interact, then they should be properly connected (know each other's
names).
We introduce a novel typing discipline (\S~\ref{sec:types}) that prevents such errors by tracking the connections
required by each procedure.
It also checks that processes store data of the correct type for the local functions that use it.
PC enjoys decidable type checking (Theorem~\ref{thm:dec}) and type inference (Theorems~\ref{thm:inference}
and~\ref{cor:inference}).

\smallpar{Endpoint Projection}
We define an EndPoint Projection (EPP) that, given a choreography, synthesises a concurrent implementation
in Procedural Processes (PP), our target process calculus (\S~\ref{sec:epp}).
PP is an abstraction of systems where concurrent processes communicate by referring to each other's 
locations or identifiers, as it happens, e.g., in MPI~\cite{MPI} or the Internet Protocol.
The synthesised code is correct by construction: it faithfully follows the behaviour
of the originating choreography (Theorem~\ref{thm:epp}).
Also, EPP is transparent, in the sense that it does not introduce any auxiliary communications or computations.
This means that a choreography always faithfully represents the actual efficiency and behaviour of the algorithm
written by the programmer.
All generated implementations are deadlock-free by construction (Corollary~\ref{cor:df-ap}).


\smallpar{Extensions}
We further discuss two extensions that allow us to write more general procedures, enhancing the expressive power of
PC: allowing parameters to be lists of processes and allowing procedure bodies to contain holes that can be filled at
runtime with arbitrary code (\S~\ref{sec:params}).
These extensions can be elegantly obtained by introducing minimal additions to the theory of PC,
demonstrating its robustness, but they are nevertheless very useful in practice.


\section{Procedural Choreographies (PC)}\label{sec:pc}
%
We begin by introducing the language model of Procedural Choreographies (\ACrich).
We focus on the synchronous semantics in this section, as the underlying theory is a bit simpler.
The asynchronous model is discussed in \S~\ref{sec:async}.

\smallpar{Syntax}
The syntax of PC is displayed in Figure~\ref{fig:pc_syntax}.
A procedural choreography is a pair $\langle\defs,C\rangle$, where $C$ is a choreography and $\defs$ is a set of
procedure definitions.
\begin{figure}[t]
\footnotesize
\begin{align*}
C ::={} & \eta;C \acspar I;C \acspar \nil
& \eta ::={} & \gencomf \acspar \gensel \acspar \genstart \acspar \gentell \\
\defs ::={} & \gendef, \defs \acspar \emptyset
& I ::={} & \gencondE \acspar \gencallP \acspar \nil
\end{align*}
\caption{Procedural Choreographies, Syntax.}
\label{fig:pc_syntax}
\end{figure}

Process names ($\pid p,\pid q,\pid r,\ldots$), identify processes that execute concurrently. Each process is
equipped with a memory cell that stores a single value of a fixed type.
Specifically, we consider a fixed set $\types$ of datatypes (numbers, lists,
etc.); each process $\pid p$ stores only values of type $T_{\pid p}\in\types$.
Statements in a choreography can either be communication actions ($\eta$) or
compound instructions ($I$), both of which can have continuations. Term $\nil$
is the terminated choreography, which we often omit in examples. We call
all terms but $\nil;C$ \emph{program terms}, or simply programs, since these
form the syntax intended for developers to use for writing programs. Term
$\nil;C$ is necessary only for the technical definition of the semantics, to
capture termination of procedure calls with continuations, and can appear only
at runtime. It is thus called a \emph{runtime term}. This distinction
plays a bigger role in \S~\ref{sec:async}, which introduces more
runtime terms to track the state of asynchronous communications.

Processes communicate via direct references (names) to each other.%
\footnote{This makes PC easy to apply to mainstream settings based on actors, objects, 
or ranks (e.g., MPI).
We could give a formulation of PC based on
standard channels (from process calculi), where each pair of process names is a channel. This is evident also
from the connection graph used in the semantics of PC, defined below.}
In a value communication $\com{\pid p.e}{\pid q.f}$, process $\pid p$ sends the result of evaluating expression $e$
to $\pid q$.
In $e$, the placeholder $\pcont$ is replaced at runtime with the data stored at process $\pid p$.
When $\pid q$ receives the value from $\pid p$, it applies to it the (total) function $f$
and stores the result.
The definition of $f$ may also access the contents of
$\pid q$'s memory.

In a selection term $\sel{\pid p}{\pid q}{l}$, $\pid p$ communicates to $\pid q$
its choice of label $l$, which is a constant. This term is intended to propagate
information on which internal choice has been made by a process to another
(see Remark~\ref{remark:selections} below).

In term $\genstart$, process $\pid p$ spawns the new process $\pid q$, which
stores data of type $T$. Process name $\pid q$ is bound in the continuation
$C$ of $\genstart;C$.

Process spawning introduces the need for mobility of process names. In
real-world systems, after execution of $\genstart$, $\pid p$ is the
only process that knows the name of $\pid q$. Any other process wanting to
communicate with $\pid q$ must therefore be first informed of its existence
(as happens, e.g., in object- and service-oriented computing~\cite{BJLZ14,GGM15}).
This is achieved with the introduction term $\gentell$, read ``$\pid p$ introduces $\pid q$ and $\pid r$''
(with $\pid p$, $\pid q$ and $\pid r$ distinct).
As its double-arrow syntax suggests, this action represents
\emph{two} communications -- one where $\pid p$ sends $\pid q$'s name to
$\pid r$, and another where $\pid p$ sends $\pid r$'s name to $\pid q$. This
will become explicit in \S~\ref{sec:epp}.

In a conditional term $\gencondE$, process $\pid p$ evaluates $e$ to choose between the possible
continuations $C_1$ and $C_2$.

The set $\defs$ contains global procedures.
Term ${X(\wtil{\pid q^T}) = C_X}$ defines a procedure $X$ with body $C_X$, which can be used anywhere in $\langle\defs,C\rangle$ -- in particular, inside $C_X$.
The names $\pids q$ are bound to $C_X$, and they are exactly the free process names in $C_X$.
Each procedure can be defined at most once in $\defs$.
Term $\gencallP$ calls (or invokes) procedure $X$ by passing $\pids p$ as parameters.
Procedure calls inside definitions must be guarded, i.e., they can only occur after a communication action.

We work up to $\alpha$-equivalence in choreographies, assuming the Barendregt convention.
Bound variables are renamed as needed when expanding procedure calls.

\begin{example}
\label{ex:ms_pc}
Recall procedure \lstinline+MS+ from our merge sort example in the Introduction (Example~\ref{ex:mergesort}).
If we annotate the parameter $\pid p$ and the started processes $\pid q_1$ and $\pid q_2$ with a type, e.g.,
$\mathbf{List}(T)$ for some $T$ (the type of lists containing elements of type $T$), then \lstinline+MS+ is a valid procedure
definition in PC, as long as we allow two straightforward syntactic conventions:
(i)~$\start{\pid p}{\wtil{\pid q^T}}$ stands for the sequence
$\start{\pid p}{\pid q_1^{T_1}};\ldots;\start{\pid p}{\pid q_n^{T_n}}$;
(ii)~a communication of the form $\com{\pid p.e}{\pid q}$ stands for $\com{\pid p.e}{\pid q.\m{id}}$, where $\m{id}$ is
the identity function: it sets the content of $\pid q$ to the value received from $\pid p$.
We adopt these conventions also in the remainder.
\end{example}

\begin{remark}[Design choices]
\label{rem:design}
We comment on two of our design choices.

The introduction action ($\gentell$) requires a three-way synchronization, and essentially performs two
communications.
The alternative development of PC with asymmetric introduction (an action
$\pid p\!: \pid q\, \code{->}\, \pid r$ whereby $\pid p$ sends $\pid q$'s name to $\pid r$,
but not conversely) is essentially the same as ours. Since in our examples we always perform introductions
in pairs, the current choice makes the presentation easier.

The restriction that each process stores only one value of a fixed type is, in practice, a minor
constraint.
As shown in Example~\ref{ex:ms_pc}, types can be tuples or lists, which mimics storing
several values.
Also, a process can create new processes with different types -- so we can encode changing the
type of $\pid p$ by having $\pid p$ create a new process $\pid p'$ and then continuing the choreography
with $\pid p'$ instead of $\pid p$.
\end{remark}

\begin{remark}[Label Selection]
\label{remark:selections}
We briefly motivate the need for selections ($\gensel$).
Consider the
choreography
$
\cond{\pid p.\m{coinflip}}{
	(\com{\pid p.\pcont}{\pid r})
}{
	(\com{\pid r.\pcont}{\pid p})
  }
$.
Here, $\pid p$ flips a coin to decide whether to send a value to $\pid r$ or to receive a value from $\pid r$.
Since processes run independently and share no data, only $\pid p$ knows which branch of the
conditional will be executed; but this information is essential for $\pid r$ to decide on its behaviour.
To propagate $\pid p$'s decision to $\pid r$, we use selections:
\[
\cond{\pid p.\m{coinflip}}{
  (\sel{\pid p}{\pid r}{\lleft}; \com{\pid p.\pcont}{\pid r})
}{
  (\sel{\pid p}{\pid r}{\lright}; \com{\pid r.\pcont}{\pid p})
  }
\]
Now $\pid r$ receives a label reflecting the choice made by $\pid p$, and can decide what to do.

This intuition is formalised by the definition of EndPoint Projection in~\S~\ref{sec:epp}.
The first choreography above is not projectable, whereas the second one is.
See also Example~\ref{ex:parallel_downloader} at the end of this section.
\end{remark}

\smallpar{Semantics}
We define a reduction semantics $\tod$ for PC, parameterised over
$\defs$.
We model the state of processes with a (total) state function $\sigma$, where $\sigma(\pid p)$ denotes the value stored in $\pid p$.
We assume that each type $T\in\types$ has a special value $\bot_T$, representing an uninitialised process state.
\begin{figure}[t]
\footnotesize
\begin{eqnarray*}
&\infer[\rname{C}{Com}]
{
  G,\gencomf;C,\sigma
  \ \tod \ 
  G, C, \sigma[\pid q \mapsto w]
}
{
  \knows{\pid p}{\pid q} &
  \eval{e[\sigma(\pid p)/\pcont]}v & \eval{f[\sigma(\pid q)/\pcont](v)}w
}
\\[1ex]
&\infer[\rname{C}{Sel}]
{
  G,\gensel;C,\sigma \ \tod \ G,C,\sigma
}
{
  \knows{\pid p}{\pid q}
}
\\[1ex]
&\infer[\rname{C}{Start}]{
  G, \genstart;C, \sigma
  \ \tod\ 
  \update{\pid p}{\pid q},C, \sigma[\pid q\mapsto\bot_T]
}{}
\\[1ex]
&\infer[\rname{C}{Tell}]{
  G, \gentell;C, \sigma
  \ \tod\ 
  \update{\pid q}{\pid r}, C, \sigma
}{
  \knows{\pid p}{\pid q}
  &
  \knows{\pid p}{\pid r}
}
\\[1ex]
&\infer[\rname{C}{Cond}]
{
  G, (\gencondE);C, \sigma \ \tod \ G, C_i\fatsemi C, \sigma
}
{
  i = 1 \ \text{if } \eval{e[\sigma(\pid p)/\pcont]}{\m{true}}, &
  i = 2 \ \text{otherwise}
}
\\[1ex]
&\infer[\rname{C}{Struct}]
{
  G, C_1, \sigma \ \tod \  G', C'_1, \sigma'
}
{
  C_1\, \precongrd \,C_2
  & G, C_2, \sigma\ \tod \  G', C'_2, \sigma'
  & C'_2 \, \precongrd\, C'_1
}
\end{eqnarray*}
\caption{Procedural Choreographies, Semantics.}
\label{pc:semantics}
\end{figure}

The semantics of PC also includes a connection graph $G$, keeping track of which processes know each other.
In the rules, $\knows{\pid p}{\pid q}$ denotes that $G$ contains an edge between $\pid p$ and $\pid q$, and
$\update{\pid p}{\pid q}$ denotes the graph obtained from $G$ by adding an edge between $\pid p$ and $\pid q$ (if
missing).

Executing a communication action $\gencomf$ in rule \rname{C}{Com} requires that: $\pid p$ and $\pid q$
are connected in $G$; $e$ is well typed; and the type of $e$ matches that
expected by the function $f$ at the receiver. The last two conditions are
encapsulated in the notation $\eval ev$, read ``$e$ evaluates to $v$''.
Choreographies can thus deadlock (be unable to reduce) because of errors in
the programming of communications; this issue is addressed by our typing
discipline in \S~\ref{sec:types}.

Rule \rname{C}{Sel} defines selection as a no-op for choreographies (see Remark~\ref{remark:selections}).

Rule \rname{C}{Start} models the creation of a process. In the reductum, the
starter and started processes are connected and can thus communicate with each
other.
Rule \rname{C}{Tell} captures name mobility, by creating a connection between
two processes $\pid q$ and $\pid r$ when they are introduced by a process $\pid p$
that is connected to both.

Rule \rname{C}{Cond} uses the auxiliary operator $\fatsemi$ to obtain a reductum
in the syntax of PC regardless of the forms of the branches $C_1$ and $C_2$ and the continuation $C$. Operator
$\fatsemi$ is defined by
$\eta\fatsemi C = \eta;C$,
$I\fatsemi C = I;C$
and
$(C_1;C_2)\fatsemi C = C_1;(C_2\fatsemi C)$.
This operator extends the scope of bound names: any name $\pid p$ bound in $C$ has its scope extended also to
$C'$.
This scope
extension is capture-avoiding, as the Barendregt convention guarantees that $\pid p$ is not used in $C'$.

Rule \rname{C}{Struct} makes use of the structural precongruence $\precongrd$,
which is defined by the rules in Figure~\ref{fig:pc_precongr}.
We write $C\equivd C'$ when $C\precongrd C'$ and $C'\precongrd C$, and denote the set of process names (free or bound)
in a choreography $C$ by $\pn(C)$.
\begin{figure}[t]
\footnotesize
\begin{eqnarray*}
&
\infer[\rname{C}{Eta-Eta}]
{
	\eta;\eta'\ \equivd \ \eta';\eta
}
{
	\pn(\eta) \cap \pn(\eta') = \emptyset
}
\quad
\infer[\rname{C}{I-I}]
{
	I;I'
	\ \equivd\ 
	I';I
}
{
	\pn(I) \cap \pn(I') = \emptyset
}
\quad
\infer[\rname{C}{I-Eta}]
{
	\eta;I
	\ \equivd\ 
	I;\eta
}
{
	\pn(I) \cap \pn(\eta) = \emptyset
}
\\[1ex]
&\infer[\rname{C}{End}]
{
	\nil;C
	\precongrd 
	C
}
{
}
\quad
\infer[\rname{C}{Unfold}] 
{
  \callP X{\pids p};C
  \ \,\precongrd\ \,
  C_X[\pids p/\pids q]\fatsemi C
}
{
  X(\wtil{\pid q^T}) = C_X \in \defs
}
\\[1ex]
&\infer[\rname{C}{Eta-Cond}]
{
	\cond{\pid p.e}{(\eta;C_1)}{(\eta;C_2)}
	\ \equivd\
	\eta;\gencondE
}
{
	\{ \pid p, \pid q\} \cap \pn(\eta) = \emptyset
}
\\[1ex]
&\infer[\rname{C}{Cond-Eta}]
{
	\cond{\pid p.e}{(C_1;\eta)}{(C_2;\eta)}
	\ \equivd\
	\gencondE;\eta
}
{
}
\\[1ex]
&\infer[\rname{C}{Cond-Cond}]
{
	\begin{array}{c}
	\cond{\pid p.e}{
 		(\cond{\pid q.e'}
 		{C_1}{C_2})
	}{
 		(\cond{\pid q.e'}
 		{C'_1}{C'_2})
	}
	\\
	\ \equivd \
	\\
	\cond{\pid q.e'}{
		(\cond{\pid p.e}
		{C_1}{C'_1})
	}{
		(\cond{\pid p.e}
		{C_2}{C'_2})
	}
	\end{array}
}
{
	\{ \pid p, \pid q \} \cap \{ \pid r, \pid s \} = \emptyset
}
\end{eqnarray*}
\caption{Procedural Choreographies, Structural precongruence $\precongr$.}
\label{fig:pc_precongr}
\end{figure}

%
Rule \rname{C}{Unfold} unfolds a procedure call, again using the $\fatsemi$ operator defined above, and rule
\rname{C}{End} is garbage collection of $\nil$.

The other rules formalise the notion of implicit parallelism anticipated in \S~\ref{sec:intro}.
Rule \rname{C}{Eta-Eta} permutes two communications performed
by processes that are all distinct, modelling that processes run independently of one another.
For example,
$\com{\pid p.\pcont}{\pid q};\com{\pid r.\pcont}{\pid s} \equivd
\com{\pid r.\pcont}{\pid s}; \com{\pid p.\pcont}{\pid q}$
because these two communications are non-interfering,
but $\com{\pid p.\pcont}{\pid q};\com{\pid q.\pcont}{\pid s} \not\equivd
\com{\pid q.\pcont}{\pid s}; \com{\pid p.\pcont}{\pid q}$:
since the second communication causally depends on the first (both involve $\pid q$).

This reasoning is extended to instructions in rule \rname{C}{I-I}; in particular,
two procedure calls can be swapped if they share no arguments.
This is sound because a procedure can only define actions for processes that are either passed as arguments or 
started inside of the procedure itself, and the latter cannot be leaked to the original call site.
Thus, any actions obtained by unfolding the first procedure call involve different processes than those obtained by unfolding the second one.
As the example below shows, calls to the same procedure can be exchanged, since $X$ and $Y$ need not be distinct.
The other rules follow similar intuitions; we omit rules \rname{C}{I-Cond} and \rname{C}{Cond-I}, analogous to \rname{C}{Eta-Cond} and \rname{C}{Cond-Eta}.
%

\begin{example}
\label{ex:impl-par}
In our merge sort example, rule \rname{C}{I-I} allows the recursive
calls \lstinline+MS<q_1>+ and \lstinline+MS<q_2>+ to be exchanged.
Furthermore, after the calls are unfolded, implicit parallelism allows their code to be interleaved 
in any way.

This example exhibits typical map-reduce behaviour: each new process receives its input, runs
independently from all others, and then sends its result to its creator.
\end{example}

\begin{example}
\label{ex:fine-grained}
A more refined example of implicit parallelism involves swapping communications from procedure
calls that share process names.
Consider the procedure
\begin{lstlisting}
auth(c,a,r,l) = c.creds -> a.rCreds;
                a.chk -> r.res; a.log -> l.app
\end{lstlisting}
Client \lstinline+c+ sends its credentials to an authentication server \lstinline+a+, which stores the result of
authentication in \lstinline+r+ and appends a log of this operation at process \lstinline+l+.
In the choreography
\lstinline+auth<c,a1,r1,l>; auth<c,a2,r2,l>+,
a client \lstinline+c+ authenticates at two different authentication servers \lstinline+a1+ and \lstinline+a2+.
After unfolding the two calls, rule \rname{C}{Eta-Eta} yields the following interleaving:
\begin{lstlisting}
c.creds -> a1.rCreds; c.creds -> a2.rCreds;
a2.chk -> r2.res; a1.chk -> r1.res;
a1.log -> l.app; a2.log -> l.app
\end{lstlisting}
Thus, the two authentications proceed in parallel. Observe that the
logging operations cannot be swapped, since they use the same logging process \lstinline+l+.
%
\end{example}

\begin{example}
\label{ex:parallel_downloader}
A more sophisticated example involves modularly composing different procedures that take multiple parameters.
Here, we write a choreography where a client $\pid c$ downloads a collection of files from a server $\pid s$.
The key idea is to download all files in parallel via streaming, by having the client and the server each
create subprocesses to handle the transfer of each file.
This allows the client to request and start downloading each file without waiting for previous
downloads to finish.
\begin{lstlisting}
par_download(c,s) = if c.more
  then c -> s [more]; c start c'; s start s';
       s: c <-> s'; c.top -> s'; pop<c>;
       c: c' <-> s'; download<c',s'>;
       par_download<c,s>; c'.file -> c.store
  else c -> s [end]
\end{lstlisting}
At the start of \lstinline+par_download+, the client $\pid c$ checks whether it wants to download more files
and informs the server $\pid s$ of the result via a label selection.
In the affirmative case, the client and the server start two subprocesses, $\pid c'$ and $\pid s'$ respectively, and
the server introduces $\pid c$ to $\pid s'$ (\lstinline+s: c <-> s'+).
The client $\pid c$ sends to $\pid s'$ the name of the file to download (\lstinline+c.top -> s'+) and removes it from its collection, using procedure \lstinline+pop+ (omitted),
afterwards introducing its own subprocess $\pid c'$ to $\pid s'$.
The file download is handled by $\pid c'$ and $\pid s'$ (using procedure \lstinline+download+), while $\pid c$ and $\pid s$ continue
operating (\lstinline+par_download<c,s>+). Finally, $\pid c'$ waits until $\pid c$ is ready to store the downloaded file.

Procedure \lstinline+download+ has a similar structure. It implements a stream where a file is
sequentially transferred in chunks from a process $\pid s$ to another process~$\pid c$.
\begin{lstlisting}
download(c,s) = if s.more
  then s -> c [more]; s.next -> c.append; pop<s>; download<c,s>
  else s -> c [end]
\end{lstlisting}

The implementation of \lstinline+par_download+ exploits implicit parallelism considerably.
All calls to \lstinline+download+ are made with disjoint sets of parameters (processes), thus they can be fully
parallelised by our semantics: many instances of \lstinline+download+ run at the same time, each one
implementing a (sequential) stream.
By implicit parallelism, we effectively end up executing many streaming behaviours in parallel.

We can even compose \lstinline+par_download+ with \lstinline+auth+, such that we execute the parallel download only if the client can successfully authenticate with an authentication server \lstinline+a+. Below, we use the shortcut $\sel{\pid p}{\pids q}{l}$ for $\sel{\pid p}{\pid q_1}{l};\dots;\sel{\pid p}{\pid q_n}{l}$.

\begin{lstlisting}
auth<c,a,r,l>; if r.ok then r -> c,s[ok]; par_download<c,s>
                       else r -> c,s[ko]
\end{lstlisting}
\end{example}



\section{Typability and Deadlock-Freedom}\label{sec:types}
%
We give a typing discipline for PC, to checks that (a) the types of functions and processes are respected by
communications and (b) processes that need to communicate are first properly introduced (or connected).
Regarding (b), two processes created independently can communicate only after they receive the names of
each other. For instance, in Example~\ref{ex:parallel_downloader}, the execution of \lstinline+download<c',s'>+ would
get stuck if \lstinline+c'+ and \lstinline+s'+ were not properly introduced in \lstinline+par_download+, since our semantics
requires them to be connected.

Typing judgements have the form $\tjudge{\Gamma}{G}{C}{G'}$, read ``$C$ is well-typed according to
the typings in $\Gamma$, and when executed from a connection graph that contains $G$ it produces a connection graph
that includes $G'$''.
Typing environments $\Gamma$ are used to track the types of processes and procedures; they are defined as:
$\Gamma ::={} \emptyset \mid \Gamma,\,\pid p:T \mid \Gamma, \, X\!:\!G\cseq G'$.
A typing $\pid p:T$ states that process $\pid p$ stores values of type $T$, and a typing $X:G\cseq G'$ records the
effect of the body of $X$ on graph~$G$.

The rules for deriving typing judgements are given in Figure~\ref{fig:pc_types}.
We assume standard typing judgements for functions and expressions, and write $\pcont:T\vdash_\types e:T$ and
$\pcont:T_1\vdash_\types f:T_2\to T_3$ meaning, respectively ``$e$ has type $T$ assuming that $\pcont$ has type $T$''
and ``$f$ has type $T_2\to T_3$ assuming that $\pcont$ has type $T_1$''.
\begin{figure}[t]
\footnotesize
\begin{eqnarray*}
&\infer[\rname{T}{End}]
{
	\Gamma;G \seq \nil \cseq G
}
{
}
\qquad
\infer[\rname{T}{Sel}]
{
  \Gamma;G\vdash\gensel;C \cseq G'
}
{
  \knows{\pid p}{\pid q} &
  \Gamma;G\vdash C \cseq G'
}
\qquad
\infer[\rname{T}{EndSeq}]
{
  \Gamma;G\seq\nil;C \cseq G'
}
{
  \Gamma;G\seq C \cseq G'
}
\\[1ex]
&
\infer[\rname{T}{Com}]
{
  \Gamma;G\seq\gencomf;C \cseq G'
}
{
  \knows{\pid p}{\pid q} &
  \Gamma\seq\pid p:T_{\pid p},\pid q:T_{\pid q} &
  \pcont:T_{\pid p}\vdash_\types e:T_1 &
  \pcont:T_{\pid q}\vdash_\types f:T_1\to T_{\pid q} &
  \Gamma;G\seq C \cseq G'
}
\\[1ex]
&\infer[\rname{T}{Cond}]
{
  \Gamma;G \seq (\gencondE);C \cseq G'
}
{
	\Gamma \seq \pid p:T
	&
  \pcont:T \seq_\types e:\m{bool} &
  \Gamma;G\seq C_i \cseq G_i &
  \Gamma;G_1 \cap G_2\vdash C \cseq G'
}
\\[1ex]
&
\infer[\rname{T}{Start}]{
  \Gamma;G\seq\genstart;C \cseq G'
}{
  \Gamma,\pid q:T;{\update{\pid p}{\pid q}} \seq C \cseq G'
}
\qquad
\infer[\rname{T}{Tell}]{
  \Gamma;G\vdash\gentell;C \cseq G'
}{
  \knows{\pid p}{\pid q} &
  \knows{\pid p}{\pid r} &
  \Gamma;\update{\pid q}{\pid r} \vdash C \cseq G'
}
\\[1ex]
&\infer[\rname{T}{Call}]
{
  \Gamma;G \seq \gencallP;C \cseq G'
}
{
	\Gamma \seq X(\wtil{\pid q^T}): G_X \cseq G'_X &
	\Gamma \seq \pid p_i:T_i &
	G_X[\pids p / \pids q] \subseteq G &
  \Gamma;G \cup (G'_X[\pids p/\pids q]) \seq C\cseq G'
}
\end{eqnarray*}
\caption{Procedural Choreographies, Typing Rules.}
\label{fig:pc_types}
\end{figure}

Verifying that communications respect the expected types is straightforward, using the connection graph $G$ to
track which processes have been introduced to each other.
In rule \rname{T}{Start}, we implicitly use the fact that $\pid q$ does not appear yet in $G$, which is another
consequence of using the Barendregt convention.
The final graph $G'$ is only used in procedure calls (rule \rname{T}{Call}).
Other rules leave it unchanged. 

To type a procedural choreography, we need to type its set of procedure definitions $\defs$.
We write $\Gamma\vdash\defs$ if: for each $X(\wtil{\pid q^T})=C_X\in\defs$, there is exactly one typing
$X(\wtil{\pid q^T}):G_X\cseq G'_X\in\Gamma$, and this typing is such that
$\Gamma,\wtil{\pid q:T},G_X\vdash C_X\cseq G'_X$.
We say that $\Gamma\vdash\langle\defs,C\rangle$ if $\Gamma,\Gamma_\defs;G_C\vdash C,G'$ for some $\Gamma_\defs$ such
that $\Gamma_\defs\vdash\defs$ and some $G'$, where $G_C$ is the full graph whose nodes are the free process names in
$C$.
The choice of $G_C$ is motivated by observing that (i)~all top-level processes should know each other and (ii)~eventual
connections between processes not occuring in $C$ do not affect its typability.

Well-typed choreographies either terminate or diverge.\footnote{Since we are
interested in communications, we assume evaluation of functions
and expressions to terminate on values with the right types (see \S~\ref{sec:related}, Faults).}
\begin{theorem}[Deadlock freedom and Subject reduction]
  \label{thm:type}
  Given a choreography $C$ and a set $\defs$ of procedure definitions, if $\Gamma\vdash\defs$ and
  $\tjudge{\Gamma}{G_1}{C}{G'_1}$ for some $\Gamma$, $G_1$ and $G'_1$, then either:
  \begin{itemize}
  \item $C\precongrd\nil$; or,
  \item for every $\sigma$, there exist $G_2$, $C'$ and $\sigma'$ such that $G_1,C,\sigma\tod G_2,C',\sigma'$ and
    $\tjudge{\Gamma'}{G_2}{C'}{G'_2}$ for some $\Gamma' \supseteq \Gamma$ and $G'_2$.
  \end{itemize}
\end{theorem}

(Proofs of theorems can be found in the Appendix.)

Checking that $\Gamma\vdash\langle\defs,C\rangle$ is not trivial, as it requires ``guessing'' $\Gamma_\defs$.
However, this set can be computed from $\langle\defs,C\rangle$, entailing type inference properties for PC.
\begin{theorem}
  \label{thm:dec}
  Given $\Gamma$, $\defs$ and $C$, $\Gamma\vdash\langle\defs,C\rangle$ is decidable.
\end{theorem}

Theorem~\ref{thm:dec} may seem a bit surprising; the key idea of its proof is that type-checking may require expanding
recursive definitions, but their parameters only need to be instantiated with process names from a finite set.
With a similar idea, we also obtain type inference.

\begin{theorem}
  \label{thm:inference}
  There is an algorithm that, given any $\langle\defs,C\rangle$, outputs:
  \begin{itemize}
  \item a set $\Gamma$ such that $\Gamma\vdash\langle\defs,C\rangle$, if such a $\Gamma$ exists;
  \item \lstinline+NO+, if no such $\Gamma$ exists.
  \end{itemize}
\end{theorem}

\begin{theorem}
  \label{cor:inference}
  The types of arguments in procedure definitions and the types of freshly created processes can be inferred
  automatically.
\end{theorem}

\begin{remark}[Inferring introductions]
\label{remark:tell}
Theorems~\ref{thm:inference} and~\ref{cor:inference} allow us to omit type annotations in choreographies,
if the types of functions and expressions at processes are known (from $\vdash_{\types}$).
Thus, programmers can write choreographies as in our examples.

The same reasoning could be adopted to infer missing introductions ($\gentell$) in a
choreography automatically, thus
lifting the programmer also from having to think about connections entirely. However, while the types inferred for a
choreography do not affect its behaviour, the placement of introductions does.
In particular, when invoking procedures one is faced with the choice of adding the necessary introductions inside the
procedure definition (weakening the conditions for its invocation) or in the code calling it
(making the procedure body more efficient).
\end{remark}


\section{Synthesising Process Implementations}\label{sec:epp}
%
We now present our EndPoint Projection (EPP), which compiles a choreography to a concurrent
implementation represented in terms of a process calculus.

\subsection{Procedural Processes (PP)}
We introduce our target process model, Procedural Processes (PP).

\smallpar{Syntax}
The syntax of PP is given in Figure~\ref{fig:pp_syntax}.
%
\begin{figure}[t]
\footnotesize
\begin{align*}
B &{}::= \asend{\pid q}{e};B \acspar \arecv{\pid p}{f};B \acspar \atells{\pid q}{\pid r};B \acspar \atellr{\pid p}{\pid r};B \acspar \asel{\pid q}{l};B
\acspar \abranch{\pid p}{\{ l_i : B_i\}_{i\in I}};B
\\
&  \acspar \nil\acspar \astart{\pid q^T}{B_2};B_1 \acspar \cond{e}{B_1}{B_2};B \acspar \gencallP ;B \acspar \nil;B\\
\pdefs &{}::= \genpdef, \pdefs \acspar \emptyset
\qquad\qquad
 N,M ::= \actor{\pid p}{v}{B} \ \acspar \ (N \parp M) \ \acspar \ \emptyN
\end{align*}
\caption{Procedural Processes, Syntax.}
\label{fig:pp_syntax}
\end{figure}

A term $\actor{\pid p}{v}{B}$ is a process, where $\pid p$ is its name, $v$ is its value, and $B$ is its behaviour.
Networks, ranged over by $N,M$, are parallel compositions of processes, where $\emptyN$ is the inactive network.
Finally, $\langle\pdefs,N\rangle$ is a procedural network, where $\pdefs$ defines the procedures that the
processes in $N$ may invoke.
Values, expressions and functions are as in PC.

A process executing a send term $\asend{\pid q}{e};B$ sends the evaluation of expression $e$ to $\pid q$, and
proceeds as $B$.
Term $\arecv{\pid p}{f};B$ is the dual receiving action: the process executing it receives a value from
$\pid p$, combines it with its value as specified by $f$, and then proceeds as $B$.
Term $\atells{\pid q}{\pid r}$ sends process name $\pid r$ to $\pid q$ and process name $\pid q$ to $\pid r$,
making $\pid q$ and $\pid r$ ``aware'' of each other.
The dual action is $\atellr{\pid p}{\pid r}$, which receives a process name from $\pid p$ that replaces the
bound variable $\pid r$ in the continuation.
Term $\asel{\pid q}{l};B$ sends the selection of a label $l$ to process $\pid q$.
Selections are received by the branching term $\abranch{\pid p}{\{ l_i : B_i\}_{i\in I}}$, which can receive a
selection for any of the labels $l_i$ and proceed according to $B_i$.
Branching terms must offer at least one branch.
Term $\astart{\pid q}{B_2};B_1$ starts a new process (with a fresh name) executing $B_2$, and proceeds in
parallel as $B_1$.
Conditionals, procedure calls, and termination are standard.
Term $\astart{\pid q}{B_2};B_1$ binds $\pid q$ in $B_1$, and $\atellr{\pid p}{\pid r};B$ binds $\pid r$ in~$B$.

\smallpar{Semantics}
The rules defining the reduction relation $\tob$ for PP are shown in Figure~\ref{fig:pp_semantics}.
As in PC, they are parameterised on the set of behavioural procedures $\mathcal B$.
%
\begin{figure}[t]
\footnotesize
\begin{eqnarray*}
&\infer[\rname{P}{Com}]
{
	\actor{\pid p}{v}{\asend{\pid q}{e};B_1}
	\ \parp\ 
	\actor{\pid q}{w}{\arecv{\pid p}{f};B_2}
	\ \tob \ 
	\actor{\pid p}{v}{B_1}
	\ \parp\ 
	\actor{\pid q}{u}{B_2}
}
{
	u = (f[w/\pcont])(e[v/\pcont])
}
\\[1ex]
&\infer[\rname{P}{Sel}]
{
	\actor{\pid p}{v}{\asel{\pid q}{l_j};B}
	\ \parp\ 
	\actor{\pid q}{w}{\abranch{\pid p}{\{ l_i : B_i\}_{i\in I}}}
	\ \tob \
	\actor{\pid p}{v}{B}
	\ \parp\ 
	\actor{\pid q}{w}{B_j}
}
{j \in I}
\\[1ex]
&\infer[\rname{P}{Cond}]
{
	\actor{\pid p}{v}{\cond{e}{B_1}{B_2}}
	\ \tob \ 
	\actor{\pid p}{v}{B_i}
}
{
	i = 1 \ \text{if } e[v/\pcont] = \m{true},\quad
	i = 2 \ \text{otherwise}
}
\\[1ex]
&\infer[\rname{P}{Start}]
{
	\actor{\pid p}{v}{
		(\astart{\pid q^T}{B_2};B_1)
	}
	\ \tob \
	\actor{\pid p}{v}{B_1[\pid q'/\pid q]}
	\ \parp\ 
	\actor{\pid q'}{\bot_T}{B_2}
}
{
\pid q' \text{ fresh}
}
\\[1ex]
&\infer[\rname{P}{Tell}]
{
	\actor{\pid p}{v}{\atells{\pid q}{\pid r};B_1}
	\ \parp\ 
	\actor{\pid q}{w}{\atellr{\pid p}{\pid r};B_2}
	\ \parp\ 
	\actor{\pid r}{u}{\atellr{\pid p}{\pid q};B_3}
	\ \tob \ 
	\actor{\pid p}{v}{B_1}
	\ \parp\ 
	\actor{\pid q}{w}{B_2}
	\ \parp\ 
	\actor{\pid r}{u}{B_3}
}
{
}
\\[1ex]
&\infer[\rname{P}{Par}]
{
	N \parp M \ \tob \ N' \parp M
}
{
	N \ \tob\  N'
}
\qquad
\infer[\rname{P}{Struct}]
{
	N \ \tob \ N'
}
{
	N \precongrb M & M\ \tob\ M' & M' \precongrb N'
}
\end{eqnarray*}
\caption{Procedural Processes, Semantics.}
\label{fig:pp_semantics}
\end{figure}
Rule \rname{P}{Com} models value communication: a process $\pid p$ executing a send action towards a process
$\pid q$ can synchronise with a receive-from-$\pid p$ action at $\pid q$; in the reductum, $f$ is used to
update the memory of $\pid q$ by combining its contents with the value sent by $\pid p$.
The placeholder $\pcont$ is replaced with the current value of $\pid p$ in $e$ (resp.\ $\pid q$ in $f$).
Rule \rname{P}{Tell} establishes a three-way synchronisation, allowing a process to introduce two others.
Since the received names are bound at the receivers, we rely on $\alpha$-conversion to make the receivers
agree on each other's name, as done in session types~\cite{HVK98}.
(Differently from PC, we do not assume the Barendregt convention here, in line with the tradition
of process calculi.)
Rule \rname{P}{Sel} is standard selection~\cite{HVK98}, where the sender process selects one of the branches offered by
the receiver.
In rule \rname{P}{Start}, we require the name of the created process to be globally fresh.
All other rules are standard.
Rule \rname{P}{Struct} uses structural precongruence $\precongrb$, which is the smallest precongruence satisfying
associativity and commutativity of parallel ($\parp$) and the rules in Figure~\ref{fig:pp_precongr}.
%
%
\begin{figure}[t]
\footnotesize
\begin{eqnarray*}
  & \infer[\rname{P}{AZero}]
     {\actor{\pid p}{v}{\nil} \ \precongrb \ \nil}{}
    \qquad
    \infer[\rname{P}{NZero}]
    {N \parp {\emptyN} \ \precongrb \ N}{}\\
  & \infer[\rname{P}{End}]{\nil;B\ \precongrb\ B}{}
    \qquad
    \infer[\rname{P}{Unfold}]
          {
            \callP X{\pids p};B
            \precongrb
            B_X[\pids p/\pids q]\fatsemi B
          }
          {
            X(\wtil{\pid q^T}) = B_X \in \pdefs
          }
\end{eqnarray*}
\caption{Procedural Processes, Structural precongruence $\precongrb$.}
\label{fig:pp_precongr}
\end{figure}

Rule \rname{P}{Unfold} expands procedure calls.
It uses again the $\fatsemi$ operator, defined as for PC
but with terms in the PP language.

\begin{remark}
  Our three-way synchronisation in rule \rname{P}{Tell} could be easily encoded with two standard two-way
  communications of names, as in the $\pi$-calculus~\cite{SW01} (see also Remark~\ref{rem:design}).
  Our choice gives a clearer formulation of EPP.
\end{remark}

\begin{example}
\label{ex:ms_proc}
We show a process implementation of the merge sort choreography in Example~\ref{ex:mergesort} from
\S~\ref{sec:intro}.
All processes are annotated with type $\mathbf{List}(T)$ (omitted); $\m{id}$ is the identity function
(Example~\ref{ex:ms_pc}).
\begin{lstlisting}
MS_p(p) = if is_small then 0
  else start q_1 |> (p?id; MS_p<q_1>; p!*);
       start q_2 |> (p?id; MS_p<q_2>; p!*);
       q_1!split_1; q_2!split_2; q_1?id; q_2?merge
\end{lstlisting}
In the next section, we show that our EPP generates this process implementation
automatically from the choreography in Example~\ref{ex:mergesort}.
\end{example}

\subsection{EndPoint Projection (EPP)}\label{sec:pc_epp}
We now show how to compile programs in PC to processes in PP.

\smallpar{Behaviour Projection}
We start by defining how to project the behaviour of a single process $\pid p$, a partial function denoted
$\epp{C}{\pid p}$.
The rules defining behaviour projection are given in Figure~\ref{fig:epp}.
\begin{figure}[t]
\footnotesize
\begin{eqnarray*}
&\epp{\gencomf;C}{\pid r} =
	\begin{cases}
		\asend{\pid q}{e};\epp{C}{\pid r} & \text{if } \pid r = \pid p \\
		\arecv{\pid p}{f};\epp{C}{\pid r} & \text{if } \pid r = \pid q \\
		\epp{C}{\pid r} & \text{otherwise}
	\end{cases}
\qquad
\epp{\gensel;C}{\pid r} =
	\begin{cases}
		\asel{\pid q}{l};\epp{C}{\pid r} & \text{if } \pid r = \pid p \\
		\abranch{\pid p}{\{ l : \epp{C}{\pid r} \}} & \text{if } \pid r = \pid q \\
		\epp{C}{\pid r} & \text{otherwise}
	\end{cases}
\\[1ex]
&\epp{\gentell;C}{\pid s} = 
\begin{cases}
	\atells{\pid q}{\pid r};\epp{C}{\pid s} & \text{if } \pid s = \pid p
	\\
	\atellr{\pid p}{\pid r};\epp{C}{\pid s} & \text{if } \pid s = \pid q
	\\
	\atellr{\pid p}{\pid q};\epp{C}{\pid s} & \text{if } \pid s = \pid r
	\\
	\epp{C}{\pid s} & \text{otherwise}
\end{cases}
\qquad
\begin{array}{l}
\epp{\gencallP;C}{\pid r} =
	\begin{cases}
		\callP{X_i}{\pids p};	\epp{C}{\pid r}
		& \text{if } \pid r = \pid p_i \\
		\epp{C}{\pid r} & \text{otherwise}
	\end{cases}
\\[1.5em]
\epp{\nil}{\pid r} = \nil
\qquad
\epp{\nil;C}{\pid r} = \epp{C}{\pid r}
\end{array}
\\[1ex]
&\epp{\gencondE;C}{\pid r} =
	\begin{cases}
		\cond{e}{\epp{C_1}{\pid r}}{\epp{C_2}{\pid r}}
		; \epp{C}{\pid r}
		& \text{if } \pid r = \pid p \\
		( \epp{C_1}{\pid r} \merge \epp{C_2}{\pid r} )
		;\epp{C}{\pid r} & \text{otherwise}
	\end{cases}
\\[1ex]
&\epp{\genstart;C}{\pid r} =
\begin{cases}
	\astart{\pid q}{\epp{C}{\pid q}};\epp{C}{\pid r}
	& \text{if } \pid r = \pid p
	\\
	\epp{C}{\pid r} & \text{otherwise}
\end{cases}
\end{eqnarray*}
\caption{Procedural Choreographies, Behaviour Projection.}
\label{fig:epp}
\end{figure}
Each choreography term is projected to the local action of the process that we are projecting.
For example, a communication term $\gencomf$ projects a send action for the sender
$\pid p$, a receive action for the receiver $\pid q$, or skips to the continuation otherwise.
The rules for projecting a selection or an introduction (name mobility) are similar.

The rule for projecting a conditional uses the partial merging operator $\merge$: $B\merge B'$
is isomorphic to $B$ and $B'$ up to branching, where the branches of $B$ or $B'$ with distinct labels are also
included.
The interesting rule defining merge is:
\begin{multline*}
  \left(\abranch{\pid p}{\{l_i:B_i\}_{i\in J}};B\right) \merge
  \left(\abranch{\pid p}{\{l_i:B'_i\}_{i\in K}};B'\right) = \\
  \abranch{\pid p}{\left(\{l_i:(B_i \merge B'_i)\}_{i\in J\cap K}\cup\{l_i:B_i\}_{i\in J \setminus K}\cup\{l_i:B'_i\}_{i\in K \setminus J}\right)};\left(B\merge B'\right)
\end{multline*}
The idea of merging comes from~\cite{CHY12}. Here, we extend it to general recursion, parametric
procedures, and process starts.
The complete definition of merging is given in the Appendix.
%
%
Merging allows the process that decides a conditional to inform other processes of its choice later on,
using selections. It is found repeatedly in most choreography models~\cite{CHY12,CDYP16,LGMZ08}.

Building on behaviour projection, we define how to project the set $\defs$ of procedure definitions.
We need to consider two main aspects.
The first is that, at runtime, the choreography may invoke a procedure $X$ multiple times, but potentially passing
a process $\pid r$ at different argument positions each time.
This means that $\pid r$ may be called to play different ``roles'' in the implementation of the procedure.
For this reason, we project the behaviour of each possible process parameter $\pid p$ as the local procedure
$X_{\pid p}$.
The second aspect is: depending on the role that $\pid r$ is called to play by the choreography, it needs
to know the names of the other processes that it is supposed to communicate with in the choreographic
procedure.
We deal with this by simply passing all arguments, which means that some of them may even be unknown to the
process invoking the procedure.
This is not a problem: we focus on typable choreographies, and typing ensures that those parameters are not
actually used in the projected procedure (so they act as ``dummies'').
We do this for clarity,
since it yields a simpler formulation of EPP. In practice,
we can annotate the EPP by analysing which parameters of each recursive definition are actually used in each of its
projections, and instantiating only those (see Appendix).
We can now define
\vspace{-3mm}
\[\epp{\defs}{}=\bigcup\left\{\epp\gendef{}\mid\gendef{}\in\defs\right\}\]
where, for $\wtil{\pid q^T} = \pid q_1^{T_1},\ldots,\pid q_n^{T_n}$,
\[\epp{\gendef}{} \ = \
\left\{ X_1(\pids q)={\epp{C}{\pid q_1}},\ldots,
X_n(\pids q)={\epp{C}{\pid q_n}} \right\} \, .
\]

\begin{definition}[EPP]
Given a procedural choreography $\langle\defs,C\rangle$ and a state $\sigma$, the endpoint projection
$\epp{\defs,C,\sigma}{}$ is the parallel composition of the processes in $C$ with all
definitions from $\defs$:
\[
\textstyle\epp{\defs,C,\sigma}{}  \ = \
\left\langle\epp\defs{},\epp{C,\sigma}{}\right\rangle \ = \
\left\langle\epp\defs{},\prod_{\pid p \in \pn(C)} \actor{\pid p}{\sigma(\pid p)}{\epp{C}{\pid p}}\right\rangle
\]
where $\epp{C,\sigma}{}$, the EPP of $C$ wrt state $\sigma$, is independent of $\defs$.
\end{definition}
Since the $\sigma$s are total, if $\epp{C,\sigma}{}$ is defined for some $\sigma$, then $\epp{C,\sigma'}{}$ is
defined also for all other $\sigma'$.
When $\epp{C,\sigma}{}=N$ is defined for any $\sigma$, we say that $C$ is \emph{projectable} and that $N$ is
the projection of $C,\sigma$.
Similar considerations apply to $\epp{\defs,C,\sigma}{}$.

\begin{example}
  The EPP of the choreography in Example~\ref{ex:mergesort} is given in Example~\ref{ex:ms_proc}.
\end{example}

\begin{example}
  We give a more sophisticated example involving merging and introductions: the projection of
  procedure \lstinline+par_download+ (Example~\ref{ex:parallel_downloader}) for process \lstinline+s+. (We omit the type annotations.)

\begin{lstlisting}
  par_download_ss(c,s) = c&{
    more: start s' |> (s?c; c?id; c?c'; download_ss<c',s'>);
          c!!s'; par_download_ss<c,s>
    end: 0                                 }
\end{lstlisting}
Observe that we call procedure \lstinline+download_ss+, as \lstinline+s'+ occurs in the position of that
procedure's formal argument \lstinline+s+.
\end{example}

\smallpar{Properties}
EPP guarantees correctness by construction: the code synthesised from a choreography follows it precisely.
\begin{theorem}[EPP Theorem]
\label{thm:epp}
If $\langle \defs, C \rangle$ is projectable, $\Gamma \seq \defs$, and
$\tjudge{\Gamma}{G}{C}{G^\ast}$, then, for all
$\sigma$:
\begin{itemize}
\item (Completeness) if $G, C,\sigma \tod G', C',\sigma'$,
  then $\epp{C,\sigma}{} \to_{\epp\defs{}} \succ \epp{C',\sigma'}{}$;
\item (Soundness) if $\epp{C,\sigma}{} \to_{\epp\defs{}} N$,
  then $G, C,\sigma \tod G', C',\sigma'$ for some $G'$, $\sigma'$ such that $\epp{C',\sigma'}{} \prec N$.
\end{itemize}
\end{theorem}
Above, the \emph{pruning relation} $\prec$ from~\cite{CHY12} eliminates the branches introduced by the
merging operator $\merge$ when they are not needed anymore to follow the originating choreography (we write $N\succ N'$
when $N' \prec N$).
Pruning does not alter reductions, since the eliminated branches are never selected~\cite{CHY12}.
Combining Theorem~\ref{thm:epp} with Theorem~\ref{thm:type} we get that the projections of typable PC terms never
deadlock.

\begin{corollary}[Deadlock-freedom by construction]
\label{cor:df-ap}
Let $N = \epp{C,\sigma}{}$ for some $C$ and $\sigma$, and assume that $\tjudge{\Gamma}{G}{C}{G'}$ for some $\Gamma$
such that $\Gamma \vdash \defs$ and some $G$ and $G'$.
Then, either:
\begin{itemize}
\item $N \precongr_{\epp{\defs}{}} \emptyN$ ($N$ has terminated);
\item or there exists $N'$ such that $N \to_{\epp\defs{}} N'$ ($N$ can reduce).
\end{itemize}
\end{corollary}

\begin{remark}[Amendment]
A choreography in PC can only be unprojectable because of unmergeable subterms.
Thus, every choreography can be made projectable by only adding label selections.
This can be formalized in an amendment algorithm, as in other choreography
languages~\cite{LMZ13,ourstuff}, reported in the Appendix.
For example, the first (unprojectable) choreography in Remark~\ref{remark:selections} can be
amended to the projectable choreography presented at the end of the same remark.

In practice, the same argument as for inferring introduction terms (Remark~\ref{remark:tell}) applies.
Although amendment allows us to write choreographies without worrying about label selections, it is useful
to give the programmer the option to place them where most convenient.
For example, consider a process $\pid p$ making an internal choice that affects processes $\pid q$ and
$\pid r$.
If one of these two processes has to perform a slower computation in response to that choice, then it makes
sense for $\pid p$ to send a label selection to it first, and only afterwards to notify the other process.
\end{remark}

\section{Asynchrony}\label{sec:async}
%
We define an alternative semantics to PC and PP, whereby communication becomes asynchronous.
We show that the results stated in the previous sections also hold for the asynchronous case.

\subsection{Asynchronous PC (aPC)}

\smallpar{Syntax} In asynchronous PC (aPC), communications are not atomically executed anymore, but consist of
multiple actions.
In particular, the sender of a message can proceed in its execution without waiting for the receiver to receive such
message.

The intuition behind aPC is that the language is extended with new runtime
terms that capture this refinement of execution steps.\footnote{Recall that runtime terms are assumed never to
  be used by programmers (like term $\nil;C$ in PC). They are used only to represent runtime states.}
At runtime communications are expanded into multiple actions.
For example, a communication $\gencomf$ expands in $\genasend$ -- a send action from $\pid p$ -- and
$\asrecv{\pid p}{\pid q}{f}{x}$ -- a receive action by $\pid q$. The process subscripts at $\bullet$ are immaterial for
the semantics of aPC, but are useful for the definition of EPP. The tag $x$ specifies that, in the original choreography,
that message from $\pid p$ should reach that receive action at $\pid q$.
Executing $\genasend$ replaces $x$ in the corresponding receive action with the actual value $v$ computed from $e$ at $\pid p$,
yielding $\asrecv{\pid p}{\pid q}{f}{v}$; executing $\asrecv{\pid p}{\pid q}{f}{v}$ updates the state of $\pid q$.

We now define aPC formally.
The new terms are given in Figure~\ref{fig:apc_syntax}, and they are all runtime terms.
%
\begin{figure}[t]
\footnotesize
\begin{align*}
\eta ::={}& \ldots
\acspar \genasend
\acspar \genarecv
\acspar \genasels
\acspar \genaselr
\acspar \genatell
\acspar \genatold
\\[1ex]
\maybe v ::={}& x \acspar v
\qquad
\maybe l ::= x \acspar l
\qquad
\maybe{\pid p} ::= x \acspar \pid p
\end{align*}
\caption{Asynchronous PC, Syntax of New Runtime Terms.}
\label{fig:apc_syntax}
\end{figure}

%
The terms follow the intuition given for value communications, extended to selection and introduction.
Label selection expands in $\genasels$ and $\genaselr$, and introduction expands in $\genatell$ and two
$\genatold$ actions.
The tags in these actions are not essential for propagating values as above, but they make the treatment of
all new actions similar.
%
%
%
Process names for the new runtime terms are defined as follows, ignoring process names in tags and
in $\bullet$ subscripts.
\begin{align*}
\pn(\genasend)={}&\pn(\genasels)=\pn(\genatell)=\{\pid p\}\\
\pn(\genarecv)={}&\pn(\genaselr)=\pn(\genatold)=\{\pid q\}
\end{align*}

\smallpar{Semantics}
The interplay between asynchronous communications and name mobility requires the connection graph to be
directed in aPC, since the processes $\pid q$ and $\pid r$ in an introduction term $\gentell$ may receive
names at different times now.
An edge from $\pid p$ to $\pid q$ in $G$, denoted $\know{\pid q}{\pid p}$, now means that $\pid p$ knows
$\pid q$'s name -- so $\pid p$ is able to send messages to $\pid q$ or to listen for incoming messages from
$\pid q$.

The semantics for aPC includes:
\begin{itemize}
\item the rules from Figure~\ref{fig:apc_red_semantics}, which are the asynchronous counterpart to \rname{C}{Com} in PC,
  and similar rules for selection and introduction;
\item rules \rname{C}{Cond} and \rname{C}{Struct} from PC;
\item rule \rname{C}{Start} from PC, where $G \cup\! \{\pid p\! \leftrightarrow\! \pid q\}$ is the graph obtained by adding
  the two directed edges between $\pid p$ and~$\pid q$ to $G$;
\item the rules defining $\precongrd$ (Figure~\ref{fig:pc_precongr}), with $\eta$ now
  ranging over all communication actions in aPC;
\item the new rules for $\precongrd$ in Figure~\ref{fig:apc_precongr}.
\end{itemize}
\begin{figure}[t]
\footnotesize
\begin{eqnarray*}
&\infer[\rname{C}{Com-S}]
{
  G,\genasend;C,\sigma
  \ \tod \ 
  G, C[v/x], \sigma
}
{
  \know{\pid p}{\pid q} &
  \eval{e[\sigma(\pid p)/\pcont]}v
}
\\[1ex]
&\infer[\rname{C}{Com-R}]
{
  G,\asrecv{\pid p}{\pid q}fv;C,\sigma
  \ \tod \ 
  G, C, \sigma[\pid q \mapsto w]
}
{
\know{\pid q}{\pid p}
&
  \eval{f[\sigma(\pid q)/\pcont](v)}w
}
\end{eqnarray*}
\caption{Asynchronous PC, Semantics of New Runtime Terms.}
\label{fig:apc_red_semantics}
\end{figure}

\begin{figure}[t]
\footnotesize
\begin{eqnarray*}
&
\infer[\rname{C}{Com-U}]
{
  \gencomf\ \precongrd\ \genasend;\asrecv{\pid p}{\pid q}fx
}
{}
\\[1ex]
&\infer[\rname{C}{Sel-U}]
{
  \gensel\ \precongrd\ \genasels;\acom{\bullet_{\pid p}}{\pid q[l]}x
}
{}\\[1ex]
&\infer[\rname{C}{Tell-U}]
{
  \gentell\ \precongrd\ \genatell;\acom{\bullet_{\pid p}.\pid r}{\pid q}{x};\acom{\bullet_{\pid p}.\pid q}{\pid r}{y}
}
{}
\end{eqnarray*}
\caption{Asynchronous PC, Structural Precongruence (New Additional Rules).}
\label{fig:apc_precongr}
\end{figure}

%
The missing new rules are given in the Appendix.
By the Barendregt convention, tags introduced by
the rules in Figure~\ref{fig:apc_precongr}
are globally fresh.
Thus, they maintain the correspondence between the value being sent and that being received.

The key to the semantics of aPC lies in the new swaps allowed by $\precongrd$, due to the definition of $\pn$ for
the new terms.
This captures the concurrency that arises from asynchronous communications.

\begin{example}
Let
$C$ be $\gencomf; \com{\pid p.e'}{\pid r.f'}
$.
To execute $C$ in aPC, we expand it:
\[
C \quad \precongrd \quad \genasend;\ \asrecv{\pid p}{\pid q}{f'}{x};\ \assend{\pid p}{e'}{\pid r}{y};\ \asrecv{\pid
p}{\pid r}{f'}{y}
\]
Using $\precongrd$, we can swap the second term to the end (it shares no process names with the subsequent
terms according to $\pn$):
\[
C \quad \precongrd \quad \genasend;\ \assend{\pid p}{e'}{\pid r}{y};\ \asrecv{\pid
p}{\pid r}{f'}{y};\ \asrecv{\pid p}{\pid q}{f'}{x}
\]
Now $\pid p$ can send both messages immediately, and $\pid r$ can receive its message before~$\pid q$.
\end{example}

\begin{example}
Due to asynchrony, a process $\pid q$ can now send a message to
another process $\pid r$ that does not yet know about $\pid q$. However, $\pid r$ is still unable to receive
it before learning $\pid q$'s name, as expected~\cite{SW01}.
Assume that $\knows{\pid p}{\pid q}$ and $\knows{\pid p}{\pid r}$ and consider the choreography:
$
C\ \defeq \ \tell{\pid p}{\pid q}{\pid r};\ \com{\pid q.\pcont}{\pid r.f}
$.
We can unfold $C$ to
\[
C \ \ \precongrd \ \  \genatell;\ \acom{\bullet_{\pid p}.\pid r}{\pid q}{x};\ \acom{\bullet_{\pid p}.\pid q}{\pid
r}{y};\ \assend{\pid q}{\pcont}{\pid r}z;\ \asrecv{\pid q}{\pid r}{f}x
\]
and, by $\precongrd$, swap the third and fourth term
\[
C \ \ \precongrd \ \ \genatell;\ \acom{\bullet_{\pid p}.\pid r}{\pid q}{x};\ \assend{\pid q}{\pcont}{\pid r}z;\
\acom{\bullet_{\pid p}.\pid q}{\pid r}{y};\ \asrecv{\pid q}{\pid r}{f}x
\]
which corresponds to an execution path where $\pid q$ sends its contents to $\pid r$ before $\pid r$ is notified of
$\pid q$'s name.
The last two actions can not be swapped, so $\pid r$ cannot access $\pid q$'s message before receiving its name.
\end{example}

To restate Theorems~\ref{thm:type}--\ref{cor:inference} for aPC, we extend our type system to the new runtime terms.
In particular, we augment contexts $\Gamma$ to contain also type declarations for tags ($x:T$) and
assignments of labels or process identifiers to tags ($x=l$ or $x=\pid p$).
The type system of aPC contains the rules for PC (Figure~\ref{fig:pc_types}) together with the new
rules given in Figure~\ref{fig:apc_red_types} and analogous ones for selections and introductions.
\begin{figure}[t]
\footnotesize
\begin{eqnarray*}
&\infer[\rname{T}{Com-S}]
{
  \Gamma;G\seq\genasend;C \cseq G'
}
{
  \know{\pid p}{\pid q} &
  \Gamma\seq\pid p:T_{\pid p} &
  \pcont:T_{\pid p}\vdash_\types e:T &
  \Gamma\oplus(x:T);G\seq C \cseq G'
}
\\[1ex]
&\infer[\rname{T}{Com-RV}]
{
  \Gamma;G\seq\asrecv{\pid p}{\pid q}fx;C \cseq G'
}
{
  \know{\pid q}{\pid p} &
  \Gamma\seq\pid q:T_{\pid q} &
  \pcont:T_{\pid q}\vdash_\types f:T\to T_{\pid q} &
  \Gamma\oplus(x:T);G\seq C \cseq G'
}
\\[1ex]
&\infer[\rname{T}{Com-RT}]
{
  \Gamma;G\seq\asrecv{\pid p}{\pid q}fv;C \cseq G'
}
{
  \know{\pid q}{\pid p} &
  \Gamma\seq\pid q:T_{\pid q} &
  \vdash_\types v:T &
  \pcont:T_{\pid q}\vdash_\types f:T\to T_{\pid q} &
  \Gamma;G\seq C \cseq G'
}
\end{eqnarray*}
\caption{Asynchronous PC, Typing Rules (New Runtime Terms).}
\label{fig:apc_red_types}
\end{figure}

Furthermore, in rule \rname{T}{End} (from PC) we require that $\Gamma$ does not contain any assertions $x:T$ or $x=t$,
meaning that there should be no dangling actions.

The context $\Gamma\oplus(x:T)$ used in the rules is defined below; the interpretation of $\Gamma\oplus(x=t)$ is
similar. The operator $\oplus$ is partial. In particular, any judgement involving
$\Gamma\oplus(x:T)$ is false if $\Gamma$ already declares $x$ with a
different type.
\[
\Gamma\oplus(x:T) =
\begin{cases}
  \Gamma,x:T & \mbox{if $x$ is not declared in $\Gamma$} \\
  \Gamma\setminus\{x:T\} & \mbox{if $x:T\in\Gamma$}\\
  \mbox{undefined} & \mbox{otherwise}
\end{cases}
\]
Thus, typing a receiving action depends on whether or not it corresponds to a message that has already been sent
according to $\Gamma$.

If we consider only program choreographies (without runtime terms), then the typing systems for aPC and PC coincide
(as the connection graph is always symmetric); in particular, Theorems~\ref{thm:dec}--\ref{cor:inference}
automatically hold for program choreographies in aPC.
The proofs of these results, as well as that of Theorem~\ref{thm:type}, can be readily adapted to aPC, as the new
cases do not pose any difficulty.

PC and aPC enjoy the following operational correspondence.
\begin{theorem}
  \label{thm:pc-vs-apc}
  For any $C$ in PC, connection graph $G$, and state $\sigma$:
  \begin{itemize}
  \item If $G,C,\sigma\tod G',C',\sigma'$ (in PC), then $G,C,\sigma\tod^* G',C',\sigma'$ (in aPC).
  \item If $G,C,\sigma\tod G',C',\sigma'$ (in aPC), then there exist $C''$, $G''$ and $\sigma''$ such that
    $G,C,\sigma\tod G'',C'',\sigma''$ (in PC) and $G',C',\sigma'\tod^* G'',C'',\sigma''$ (in aPC).
  \end{itemize}
\end{theorem}
In particular, the sequence of messages from a process $\pid p$ to another process $\pid q$ is received by $\pid q$
in exactly the same order as it is sent by $\pid p$. Thus, our semantics matches an interpretation of
asynchrony supported by queues between processes, as we make evident in the asynchronous
version of PP below.

\subsection{Asynchronous PP (aPP)}

\smallpar{Syntax} The asynchronous variant of PP, aPP, is easier to define, as the underlying language is almost
unchanged. The only syntactic difference is that processes now have the form $\actor[\rho]{\pid p}vB$, where $\rho$ is a
queue of incoming messages. A message is a pair $\langle\pid q,m\rangle$, where $\pid q$ is the sender process
and $m$ is a value, label, or process identifier.

\smallpar{Semantics} The semantics of aPP consists of rules \rname{P}{Cond}, \rname{P}{Par}, \rname{P}{Cond},
and \rname{P}{Start} from PP (Figure~\ref{fig:pp_semantics}) together with the rules given in
Figure~\ref{fig:app_red_semantics} and analogous variants for selection and introduction (see Appendix).
In the reductum of \rname{P}{Start}, the newly created process is initialized with an empty queue.
Structural precongruence for aPP is defined exactly as for PP.
%
\begin{figure}[t]
\footnotesize
\begin{eqnarray*}
&\infer[\rname{P}{Com-S}]
{
  \actor[\rhop]{\pid p}{v}{\asend{\pid q}{e};B_{\pid p}} \parp \actor[\rhoq]{\pid q}w{B_{\pid q}}
  \ \tob \ 
  \actor[\rhop]{\pid p}{v}{B_{\pid p}} \parp \actor[\rho'_{\pid q}]{\pid q}w{B_{\pid q}}
}
{
  \rho'_{\pid q}=\rhoq\cdot\langle\pid p,e{[v/\pcont]}\rangle
}
\\[1ex]
&\infer[\rname{P}{Com-R}]
{
  \actor[\rhoq]{\pid q}{w}{\arecv{\pid p}{f};B}
  \ \tob \ 
  \actor[\rho'_{\pid q}]{\pid q}{u}{B}
}
{
  \rhoq\preceq\langle\pid p,v\rangle\cdot\rho'_{\pid q} &
  u = (f[w/\pcont])(v)
}
\end{eqnarray*}
\caption{Asynchronous Procedural Processes, Semantics (New Rules).}
\label{fig:app_red_semantics}
\end{figure}
%

%
We write $\rho \cdot \langle\pid q,m\rangle$ to denote the queue obtained by appending message $\langle\pid q,m\rangle$
to $\rho$, and $\langle\pid q,m\rangle \cdot \rho$ for the queue with $\langle\pid q,m\rangle$ at the head and $\rho$
as tail.
We simulate having one separate queue for each other process by allowing incoming messages from different
senders
to be exchanged, which we represent using the congruence $\rho\preceq\rho'$ defined by the rule
$\langle\pid p,m\rangle \cdot \langle\pid q,m'\rangle \preceq \langle\pid q,m'\rangle
\cdot \langle\pid p,m\rangle$ if $\pid p \neq \pid q$.

All behaviours of PP are valid also in aPP.
\begin{theorem}
  \label{thm:pp-vs-app}
  Let $N$ be a PP network.
  If $N\tob N'$ (in PP), then $N_{[]}\tob^*N'_{[]}$ (in aPP), where $N_{[]}$ denotes the
  asynchronous network obtained by adding an empty queue to each process.
\end{theorem}

The converse is not true, so the relation between PP and aPP is not so strong as
in Theorem~\ref{thm:pc-vs-apc} for PC and aPC. This is because of deadlocks: in
PP, a communication action can only take place when the sender and receiver are
ready to synchronize; in aPP, a process can send a message to another process,
even though the intended recipient is not yet able to receive it.
For example, the network $\actor{\pid p}{v}{\asend{\pid q}\pcont}\parp\actor{\pid q}{w}{\nil}$
is deadlocked in PP but not in aPP.
%
%

\subsection{EndPoint Projection (EPP) from aPC to aPP}

Defining an EPP from aPC to aPP requires extending the previous definition with clauses for the new runtime terms.
The most interesting part is generating the local queues for each process: when compiling
$\com{\bullet_{\pid p}.v}{\pid q.f}$, for example, we need to add $\langle \pid p, v \rangle$ at the top of $\pid q$'s
queue.
However, if we follow this intuition, Theorem~\ref{thm:epp} no longer holds for arbitrary aPC choreographies.
This happens because we can write choreographies that use runtime terms in a ``wrong'' way, and these are not
correctly compiled to aPP.
Consider the choreography $C=\com{\pid p.1}{\pid q.f};\asrecv{\pid p}{\pid q}{f}{2}$.
If we naively project it for some $\sigma$, we obtain
$
\actor[{[]}]{\pid p}{\sigma(\pid p)}{\asend{\pid q}1}\parp\actor[\langle\pid p,2\rangle]{\pid q}{\sigma(\pid q)}{\arecv{\pid p}f;\arecv{\pid p}f}\,,
$
where $[]$ is the empty queue, and $\pid q$ will receive $2$ before $1$.

To avoid this undesired behaviour, we recall that runtime terms should not be used by
the programmer and restrict ourselves to \emph{well-formed} choreographies: those that can arise from executing a
choreography that does not contain runtime terms (i.e., a program).
\begin{definition}[Well-formedness]
  A choreography $C$ in aPC is \emph{well-formed} if $C^R\fatsemi C^{PC} \precongr C$ where:
  \begin{itemize}
  \item $\precongr$ is structural precongruence without rule \rname{C}{Unfold};
  \item $C^R$ only contains unmatched and instantiated receive actions: $\asrecv{\pid p}{\pid q}{f}v$,
    $\acom{\bullet_{\pid p}}{\pid q[l]}l$, $\acom{\bullet_{\pid p}.\pid r}{\pid q}{\pid r}$;
  \item $C^{PC}$ is a PC choreography.
  \end{itemize}
  A procedure $\gendef$ is \emph{well-formed} if $C$ does not contain any runtime actions.
  A procedural choreography $\langle \defs,C \rangle$ is well-formed if all procedures in $\defs$ are well-formed and
$C$ is well-formed.
\end{definition}
Well-formedness is decidable, since the set of choreographies equivalent up to $\precongr$ is decidable.
More efficiently, one can check that $C$ is well-formed by swapping all runtime actions to the beginning and
folding all paired send/receive terms.
Furthermore, choreography execution preserves well-formedness. The problematic choreography described above is not
well-formed.

\begin{definition}[EPP from aPC to aPP]
  Let $\langle\defs,C\rangle$ be in aPC and $\sigma$ be a state.
  The EPP of $\langle\defs,C\rangle$ and $\sigma$ is defined as
  \[
  \textstyle\epp{\defs,C,\sigma}{}  \ = \
  \left\langle\epp\defs{},\epp{C,\sigma}{}\right\rangle \ = \
  \left\langle\epp\defs{},\prod_{\pid p \in \pn(C)} \actor[\something{C}{\pid p}]{\pid p}{\sigma(\pid p)}{\epp{C}{\pid p}}\right\rangle
  \]
  where $\epp{C}{\pid p}$ is defined as for PC (Figure~\ref{fig:epp}) with the extra rules in Figure~\ref{fig:aepp_red}
  and $\something{C}{\pid p}$ is defined by the rules in Figure~\ref{fig:aepp_red_state}.
  The similar rules for selection and introduction are defined in the Appendix.
\end{definition}
\begin{figure}[t]
\footnotesize
\begin{eqnarray*}
&\epp{\genasend;C}{\pid r}=
  \begin{cases}
    \asend{\pid q}{e};\epp{C}{\pid r} & \text{if } \pid r = \pid p \\
    \epp{C}{\pid r} & \text{otherwise}
  \end{cases}
\qquad
\epp{\genarecv;C}{\pid r}=
  \begin{cases}
    \arecv{\pid p}{f};\epp{C}{\pid r} & \text{if } \pid r = \pid q \\
    \epp{C}{\pid r} & \text{otherwise}
  \end{cases}
\end{eqnarray*}
\caption{Asynchronous PC, Behaviour Projection (New Rules).}
\label{fig:aepp_red}
\end{figure}

\begin{figure}[t]
\footnotesize
\begin{eqnarray*}
&\something{\asrecv{\pid p}{\pid q}fv;C}{\pid r}=
  \begin{cases}
    \langle\pid p,v\rangle\cdot\something{C}{\pid r} & \text{if } \pid r = \pid q \\
    \something{C}{\pid r} & \text{otherwise}
  \end{cases}
\qquad
\begin{array}l
\something{\gencond;C}{\pid r} = \something{C_1}{\pid r}\cdot\something{C}{\pid r}\\[1ex]
\something{\eta;C}{\pid r} = \something{I;C}{\pid r} = \something{C}{\pid r}
\end{array}
\end{eqnarray*}
\caption{Asynchronous PC, State Projection.}
\label{fig:aepp_red_state}
\end{figure}

In the last case of the definition of $\something{C}{\pid p}$ (bottom-right), $\eta$ and $I$ range over all cases that
are not covered previously.
The rule for the conditional may seem a bit surprising: in the case of projectable choreographies, mergeability
and well-formedness together imply that unmatched receive actions at a process must occur in the same order in both
branches.
Note that projection of choreographies with ill-formed runtime terms (e.g.,
$\acom{\bullet_{\pid p}}{\pid q[l]}{l'}$ with $l\neq l'$, which cannot occur in a well-formed choreography) is
not defined.

With this definition, we can restate Theorem~\ref{thm:epp} for aPC and aPP.
\begin{theorem}[Asynchronous EPP Theorem]
  \label{thm:aepp}
  If $\langle \defs, C \rangle$ in aPC is projectable and well-formed,
  $\Gamma \seq \defs$, and $\tjudge{\Gamma}{G}{C}{G^\ast}$, then, for all
  $\sigma$:
  \begin{itemize}
  \item (Completeness) if $G, C,\sigma \tod G', C',\sigma'$,
    then $\epp{C,\sigma}{} \to_{\epp\defs{}} \succ \epp{C',\sigma'}{}$;
  \item (Soundness) if $\epp{C,\sigma}{} \to_{\epp\defs{}} N$,
    then $G, C,\sigma \tod G', C',\sigma'$ for some $G'$, $\sigma'$ such that $\epp{C',\sigma'}{} \prec N$.
  \end{itemize}
\end{theorem}
As a consequence, Corollary~\ref{cor:df-ap} applies also to the asynchronous
case: the processes projected from aPC into aPP are deadlock-free.




\section{Language Extensions}\label{sec:params}
%
We sketch two simple extensions that enhance the expressivity of (synchronous and asynchronous) PC without affecting its
underlying theory: they only require simple modifications, and all the
results from the previous sections still hold.

\smallpar{Parameter Lists}
Our first extension allows procedure parameters to be
lists of processes, on which procedures can then act uniformly by recursion.
Formally, a procedure parameter can now be either a process or a list of processes all with the same type.
We restrict the usage of such lists to the arguments of procedure calls; however, they may be manipulated by means
of pure total functions that take a list as their only argument.
For example, if procedure $X$ takes a list $P$ as an argument, then $X$'s body may call other procedures on $P$,
$\mathsf{head}(P)$ or $\mathsf{tail}(P)$, but it may not include a communication involving $\mathsf{head}(P)$ and another
process.

The semantics is extended with a new garbage collection rule: $\callP X{[]}\precongrd\nil$, i.e., calling a
procedure with an empty list is equivalent to termination ($\nil$).
In order to type procedures that take process lists as arguments, in a typing $X\!:\!G\cseq G'$ we
allow the vertices of $G$ and $G'$ to be not only processes, but also (formal) lists.
This requires adjusting the premise of \rname{T}{Call} to the case when some of the arguments are lists. For when a
concrete list is given as argument, we update
the premise $G_X[\pids p/\pids q]\subseteq G$ to check
that: if $G_X$ contains an edge between two process lists $P$ and $Q$, then $G$ must contain
edges between $\pid p$ and $\pid q$ for each $\pid p$ in the list supplied for $P$ and each $\pid q$ in the
list supplied for $Q$.
The interpretation of an edge between an argument process and an argument list is similar.

We extend EPP to this enriched language by similarly extending the syntax of PP and adding the rule $\actor{\pid q}v{\gencallP}\precongrd\actor{\pid q}v\nil$
if $\pid q\not\in\pids p$.
Although we cannot call procedures with parameter lists based on the result of a conditional without changing PP
(in particular, due to the definition of merge), in this language we can write more sophisticated
examples.
In~\cite{CM16a}, we show how to implement Quicksort, Gaussian elimination and the Fast Fourier Transform in PC,
and discuss how implicit parallelism yields efficient implementations of these algorithms.




\smallpar{Procedures with Holes}
The second extension allows the presence of holes in procedure definitions.
Holes are denoted $\Box_h$, where $h$ is a unique name for the hole in the procedure.
Procedure calls are then allowed to specify a choreography to inject in each hole by means of the syntax
$\gencallP\ \m{with}\ [h_i \mapsto C_i]_i$ where each $C_i$ is a choreography (in particular, it
can be another procedure call) with $\pn(C_i) \subseteq \pids p$. We assume unspecified holes to be filled with $\nil$.

The typing of procedure definitions is extended to
$X(\wtil{\pid q^T}):G\cseq G' \ \m{with}\ \{h_i:G_i\}_i$, where $G_i$ is the guaranteed connection graph
when reaching hole $\Box_{h_i}$.
The remaining adaptations (syntax and semantics of PC and PP, and EPP) are trivial.

\begin{example}
Holes enable higher-order composition of choreographies. Consider the following
\lstinline+exchange+ procedure:
\begin{lstlisting}
exchange(p,q,r,s) = p.item -> q.get; \Box_h1; r.item -> s.get; \Box_h2
\end{lstlisting}
We can either leave the holes empty, or use them to perform some extra actions,
e.g., inserting a payment:
\begin{lstlisting}
exchange<p,q,r,s> with h1 \mapsto pay<p,q>, h2 \mapsto pay<r,s>
\end{lstlisting}
We can even insert new causal dependencies between the two communications in
\lstinline+exchange+, e.g., using \lstinline+q+ as a broker:
\begin{lstlisting}
exchange<p,q,r,s> with h1 \mapsto q.item -> r.get
\end{lstlisting}
\end{example}

Using holes, we can also define a general-purpose iterator (see Appendix).


\section{Discussion, Future and Related Work}\label{sec:related}
%

\smallpar{Design of PC}
PC is part of a long-term research effort motivated by empirical studies on concurrency bugs, such as 
the taxonomies~\cite{LPSZ08} and~\cite{LLLG16}. They distinguish between deadlock and non-deadlock bugs.
In distributed systems, both are mostly caused by the wrong composition of
different protocols, rather than the wrong implementation of single protocols.
Deadlocks are prevented in PC (Corollary~\ref{cor:df-ap}). Non-deadlock bugs are
more subtle, and typically due to the informal specification (or lack thereof) of the
programmer's intentions wrt protocol composition, which leads to unexpected bad
executions. In PC, the programmer's intentions are formalised as a choreography
and the composition of procedures in code. Synthesis (EPP) guarantees that these
intentions are respected (Theorem~\ref{thm:epp}). Typical non-deadlock bugs
include lack of causality between messages, or between an
internal computation and a message~\cite{LLLG16}. For example, a client may forget to wait for an authentication 
protocol to complete and include the resulting security key in its request to a server in
another protocol. We can formalise the correct scenario in PC using our
\lstinline+auth+ procedure from Example~\ref{ex:fine-grained} (\lstinline+c+ is
the client and \lstinline+s+ is the server):
\begin{lstlisting}
auth<c,a,r,l>;
if r.ok then r -> c,s[ok];
             r.key -> c.addKey; request<c,s>
        else r -> c,s[ko]
\end{lstlisting}
%

There is still ground to be covered to support all the features (e.g., pre-emption~\cite{LLLG16}) used in some 
of the programs 
included in these empirical studies. Thus, this kind of studies will also
be important for future developments of PC.

\smallpar{Multiparty Session Types} In Multiparty Session Types (MPST)~\cite{HYC16}, global types are 
specifications of single protocols, used for verifying the code of
manually-written implementations in process models.
Global types are similar to a simplified fragment of PC, obtained (among other restrictions) by replacing expressions 
and functions with constants (representing types), removing process creation (the processes are fixed), and restricting 
recursion to parameterless tail recursion.

MPST leaves the composition of protocols as an implementation detail to the implementors of processes.
As a consequence, protocol compositions may lead to deadlocks, unlike in PC. We illustrate this key difference with 
an example, using our syntax. Consider the
protocols $X(\pid r,\pid s) = \com{\pid r.e}{\pid s.f}$ and $Y(\pid r',\pid s') = \com{\pid r'.e'}{\pid 
s'.f'}$, and assume that we want to compose their instantiations $\callP{X}{\pid p, \pid q}$ and $\callP{Y}{\pid q,\pid 
p}$.
According to MPST, a valid implementation (paraphrased in PP) is
$\actor{\pid p}{v}{\arecv{\pid q}{f'}; \asend{\pid q}{e}} \parp \actor{\pid q}{v}{\arecv{\pid p}{f}; \asend{\pid 
p}{e'}}$. Although this network is obviously deadlocked, this is not detected by MPST because the interleaving of the 
two protocols is not checked. In PC, we can only obtain correct implementations, because compositions are 
defined at the level of choreographies, e.g., $\callP{X}{\pid p, \pid q}; \callP{Y}{\pid q,\pid 
p}$ or $\callP{Y}{\pid q,\pid p}; \callP{X}{\pid p, \pid q}$.

The authors of~\cite{CDYP16} augment MPST with a type system that prevents protocol compositions that may lead 
to deadlocks, but their approach is too restrictive when actions from different protocols
are interleaved. Consider the following example, reported as untypable but deadlock-free in~\cite{CDYP16} (adapted 
to PC).
\begin{lstlisting}
X(p,q) = p.* -> q; \Box_h1; q.* -> p; p.* -> q; \Box_h2

X<p,q> with h1 \mapsto q.* -> p; p.* -> q, h2 \mapsto q.* -> p
\end{lstlisting}
This code is both typable and projectable in PC, yielding a correct implementation.

Another technique for deadlock-freedom in MPST and similar models is to restrict connections
among processes participating in different protocols to form a tree~\cite{CD10,CLMSW16,CMSY15}.
PC is more expressive, since connections can form an arbitrary graph.

Recent work investigated how to extend MPST to capture
protocols where the number of participants in a session is fixed only 
at runtime~\cite{YDBH10}, or can grow during execution~\cite{DY11}. These results use ad-hoc primitives 
and ``middleware'' terms in the process model, e.g., for tracking the number of participants in a
session~\cite{DY11}. Such machinery is not needed in PC.
The authors of~\cite{DH12} propose a theory of nested MPST, which
partially recalls our notion of parametric procedures. Differently from PC, procedures are
invoked by a coordinator (requiring extra communications), and 
compositions of such nested types can deadlock.


\smallpar{Choreographic Programming} None of the examples in this
work can be written in previous models for choreographic programming, which
lack full procedural abstraction. As far as we know, this is also
the first work exploring choreographic programming in a
general setting, instead of in the context of web services~\cite{CHY12,CM13,wscdl}.

The models in~\cite{CHY12,CM13} are based on choreographic 
programming and thus support deadlock-freedom even when multiple protocols are used.
However, these models do not support general recursion and parametric procedures, both of which are novel in this work.
Composition and reusability in these models are thus limited: procedures cannot have 
continuations; there can only be a limited number of protocols running at 
any time (modulo dangling asynchronous actions); and the process names used in a procedure are statically determined, 
preventing reuse.
In PC, all these limitations are lifted.
%
%

Asynchronous communications in choreographic programming were addressed in~\cite{CM13} using an ad-hoc transition rule.
Our approach in aPC, which builds on implicit parallelism, is simpler and yields an EPP Theorem (Theorem~\ref{thm:aepp}) 
that states a lockstep correspondence between the programmer's choreography and the
synthesised implementation. Instead, in~\cite{CM13}, EPP guarantees only a weaker
and more complex confluence correspondence.

Holes in PC recall adaptation scopes for choreographies in~\cite{PGGLM15}, used to pinpoint 
where runtime adaptation of code can take place. However, that work deals only with choreographies with
a finite number of processes.


A major distinguishing feature of PC is the management of connections among
processes, using graphs that are manipulated at runtime. Channel delegation in choreographies~\cite{CM13} behaves in a 
similar way, but it is less expressive: a process that introduces two other 
processes cannot communicate with them thenceforth without requiring a new auxiliary channel.

\smallpar{Other paradigms} Reo is another approach that, like PC, disentangles computation 
from communication and focuses on the programming of interactions, rather than the actions that implement them. In Reo, 
protocols are obtained by composing (graphical) connectors, which mediate 
communications among components~\cite{A04}. As such, Reo produces rather different and more abstract artifacts than PC, 
which require additional machinery to be implemented (a compiler from Reo to Java uses constraint 
automata~\cite{ABRS04}).
Programs in PC can be implemented rather directly via standard send and receive primitives (which is typical of
choreographies~\cite{CM13}). PC also supports dynamic creation of 
processes and guarantees deadlock-freedom by construction, through its simple type system and EPP.
Reo supports many interesting connectors and compositions (e.g., multicast), which
would be interesting to investigate in PC.

In a different direction, there have been proposals of integrating communication protocols (given as session 
types) with functional programming~\cite{VGR06}. These approaches allow parallel computations of functions to 
synchronise and exchange values. As for MPST, protocol compositions are left to the programmer, and 
the same limitations discussed for MPST apply.

\smallpar{Local concurrency}
Some choreography models support an explicit parallel operator, e.g., $C \parp C'$.
We chose not to present it here to keep our model simple, as most behaviours of $C \parp C'$ are captured 
by implicit parallelism and asynchrony in PC. For example, $\callP{X}{\pid p,\pid q} \parp \callP{Y}{\pid r,\pid s}$ 
is equivalent to $\callP{X}{\pid p,\pid q}; \callP{Y}{\pid r,\pid s}$ in PC (Rule~\rname{C}{I-I}, 
Figure~\ref{fig:pc_precongr}). However, when some processes names are shared
in the parallel composition, then the 
parallel operator is more expressive, because it 
allows for \emph{local concurrency}.
As an example, consider
$\callP{X}{\pid p,\pid q} \parp \callP{Y}{\pid p,\pid s}$:
here, $\pid p$ is doing more actions concurrently.
In PC, we can encode this by spawning new processes, e.g.~our example would become
 $\start{\pid p}{\pid r}; \tell{\pid p}{\pid r}{\pid s}; \callP{X}{\pid p,\pid q}; \callP{Y}{\pid r,\pid s}$.
The price we pay is the extra communications to introduce $\pid r$ ($\pid p$'s new 
process) to $\pid s$ (plus, we may need to send $\pid p$'s value to $\pid r$, and maybe send $\pid r$'s value to $\pid p$ after $Y$ terminates).

If local concurrency in PC is desired, there are different options for introducing it.
A straightforward one is to adapt the parallel operator from~\cite{CHY12,HYC16}, which, however,
abstracts from how local concurrency is implemented.
A more interesting 
strategy would be to support named local threads in processes. Threads would share the process location, such that 
they can receive at the same name without the need to be introduced.
Then, implicit parallelism would 
allow for swapping actions at threads (not just processes) with distinct names. Threads at the 
same process should not receive from the same sender process at the same time, to avoid 
races, which can be prevented via typing (cf.~\cite{CHY12}).
Integration with separation logic~\cite{R02} would also be interesting, to allow threads also to safely share data.
%

\smallpar{Sessions, Connections, and Mobility} All recent theories based on
session types (e.g.,~\cite{CHY12,CM13,CMSY15,CDYP16,HYC16}) assume that all processes in a session (a protocol 
execution) have
a private full-duplex channel to communicate with each other in that session.
This amounts to requiring that the graph of
connections among processes in a protocol is always complete. This assumption is
not necessary in PC, since our semantics and typing require only the processes
that actually communicate to be connected. This makes PC a suitable model
for reasoning about different kinds of network topologies. In the future, it
would be interesting to see whether our type system and connection graphs can be
used to enforce pre-defined network structures (e.g., hypercubes or
butterflies), making PC a candidate for the programming of choreographies that
account for hardware restrictions.

Another important aspect of sessions is that each new protocol execution
requires the creation of a new session (and all the connections among its
participants!). By contrast, protocol executions (running procedures) in PC can
reuse all available connections --
allowing for more efficient implementations.
%
We use this feature in our parallel downloader example
(Example~\ref{ex:parallel_downloader}), where \lstinline+c'+ can still
communicate with \lstinline+c+ even after the latter introduced it to
\lstinline+s'+. Also, aPC supports for the first time the asynchronous
establishment of new connections among processes.

The standard results of communication safety found in session-typed calculi
can be derived from our EPP Theorem (Theorem~\ref{thm:epp}), as discussed
in~\cite{CM13}.

\smallpar{Faults}
Some interesting distributed bugs~\cite{LLLG16} are triggered by unexpected fault conditions
at nodes, making such faults an immediate candidate for
the future developments of PC. Useful inspiration to this aim may be provided by~\cite{CGY16}.

We have also abstracted from faults and divergence of internal computations: in
PC, we assume that all internal computations terminate successfully. If we relax
these conditions, deadlock-freedom can still be achieved simply by using
timeouts and propagating faults through communications.

\smallpar{Integration} Sometimes, the code synthesised from a choreography has
to be used in combination with legacy process code. For example, we may want to
use an existing authentication server in our \lstinline+auth+ procedure in
Example~\ref{ex:fine-grained}. This issue is addressed in~\cite{MY13} by using a type
theory based on sessions. We leave an adaptation of this idea in our setting to
future work.

\clearpage

\bibliographystyle{plain}
\bibliography{biblio}

\clearpage\appendix
\section{Detailed definitions and proofs}
We report on definitions and proofs that were omitted from the main part of this paper, as well
as some more detailed examples.

\subsection{Type checking and type inference}
We start with some technical lemmas about typing.

\begin{lemma}[Monotonicity]
  \label{lem:mon}
  Let $\Gamma$ and $\Gamma'$ be typing contexts with $\Gamma\subseteq\Gamma'$, $G_1$,
  $G_1'$ and $G$ be connection graphs such that $G_1\subseteq G$, and $C$ be a
  choreography.
  If $\tjudge\Gamma{G_1}C{G'_1}$, then $\tjudge{\Gamma'}{G}C{G\cup G'_1}$.
\end{lemma}
\begin{proof}
  Straightforward by induction on the derivation of $\tjudge\Gamma{G_1}C{G'_1}$.
\end{proof}

\begin{lemma}[Sequentiality]
  \label{lem:seq}
  Let $\Gamma$ be a typing context, $G_1$, $G_1'$, $G_2$ and $G_2'$ be connection graphs
  such that $G_2\subseteq G'_1$, and $C_1$, $C_2$ be choreographies.
  If $\tjudge\Gamma{G_1}{C_1}{G'_1}$ and $\tjudge\Gamma{G_2}{C_2}{G'_2}$, then
  $\tjudge\Gamma{G_1}{C_1\fatsemi C_2}{G'_1\cup G'_2}$.
\end{lemma}
\begin{proof}
  Straightforward by induction on the derivation of $\tjudge\Gamma{G_2}{C_2}{G'_2}$.
\end{proof}

\begin{lemma}[Substitution]
  \label{lem:sub}
  Let $\Gamma$ be a typing context, $G$ and $G'$ be connection graphs, and $C$ be a
  choreography.
  Let $\pids p$ be a set of process names that are free in $C$ and $\pids q$ be a set of
  process names that do not occur (free or bound) in $C$.
  If $\tjudge\Gamma GC{G'}$, then
  $\tjudge{\Gamma[\pids p/\pids q]}{G[\pids p/\pids q]}{C[\pids p/\pids q]}{G'[\pids p/\pids q]}$.
\end{lemma}
\begin{proof}
  Straightforward by induction on the derivation of $\tjudge\Gamma GC{G'}$, as all typing
  rules are valid when substitutions are applied.
\end{proof}

We are now ready to start proving Theorem~\ref{thm:type}.
The following lemma takes care of the base cases, and is required for one of the inductive
steps.

\begin{lemma}
  \label{lem:type}
  Let $\Gamma$ be a set of typing judgements, $\defs$ a set of procedure definitions,
  $G_1$ and $G_1'$ connection graphs, and $C$ a choreography that does not start with
  $\nil$ or a procedure call.
  Assume that $\Gamma\vdash\defs$ and $\tjudge\Gamma{G_1}C{G'_1}$.
  For every state $\sigma$, there exist $\Gamma'$, $\sigma'$, $C'$, $G_2$ and
  $G_2'$ such that $G_1,C,\sigma\tod G_2,C',\sigma'$ and
  $\tjudge{\Gamma'}{G_2}{C'}{G'_2}$.
\end{lemma}
\begin{proof}
  By case analysis on the last step of the proof of $\tjudge\Gamma{G_1}C{G'_1}$.
  By hypothesis, this proof cannot end with an application of rules \rname{T}{End},
  \rname{T}{EndSeq} or \rname{T}{Call}; we detail all cases for completeness, but the only
  non-trivial one is the last.
  \begin{itemize}
  \item\rname{T}{Start}: then $C$ is $\genstart;C^\circ$ and by hypothesis
    $$\tjudge{\Gamma,q:T}{\update[G_1]pq}{C^\circ}{G'_1}\,.$$
    Since $G_1,\genstart;C,\sigma\tod\update[G_1]pq,C^\circ,\sigma[\pid q\mapsto\bot_T]$
    by rule \rname{C}{Start}, taking $\Gamma'=\Gamma,\pid q:T$,
    $\sigma'=\sigma[\pid q\mapsto\bot_T]$, $C'=C^\circ$, $G_2=\update[G_1]pq$ and
    $G'_2=G'_1$ establishes the thesis.
  \item\rname{T}{Com}: then $C$ is $\gencomf;C^\circ$ and by hypothesis $\knows[G_1]pq$,
    $f[\sigma(\pid q)/\pcont](e[\sigma(\pid p)/\pcont])$ is a valid expression of type
    $T_q$, and $\tjudge\Gamma{G_1}C{G'_1}$.
    Then all the preconditions of \rname{C}{Com} are met, so
    taking $\Gamma'=\Gamma$,
    $\sigma'=\sigma[\pid q\mapsto f[\sigma(\pid q)/\pcont](e[\sigma(\pid p)/\pcont])]$,
    $C'=C^\circ$, $G_2=G_1$ and $G'_2=G'_1$ establishes the thesis.
  \item\rname{T}{Sel}: then $C$ is $\gensel;C^\circ$ and by hypothesis $\knows[G_1]pq$ and
    $\tjudge\Gamma{G_1}C{G'_1}$.
    By \rname{C}{Sel}, $G_1,\gensel;C^\circ,\sigma\tod G_1,C^\circ,\sigma$, so
    taking $\Gamma'=\Gamma$, $\sigma'=\sigma$, $C'=C^\circ$, $G_2=G_1$ and
    $G'_2=G'_1$ again establishes the thesis.
  \item\rname{T}{Tell}: then $C$ is $\gentell;C^\circ$ and by hypothesis both
    $\knows[G_1]pq$, $\knows[G_1]pr$, and $\tjudge\Gamma{\update[G_1]qr}{C^\circ}{G'_1}$.
    Since the preconditions of rule \rname{C}{Tell} are met, by taking
    $\Gamma'=\Gamma$, $\sigma'=\sigma$, $C'=C^\circ$, $G_2=\update[G_1]qr$ and
    $G'_2=G'_1$ establishes the thesis.
  \item\rname{T}{Cond}: then $C$ is $\gencondE;C^\circ$ and by hypothesis
    $e[\sigma(\pid p)/\pcont]$ is a valid Boolean expression,
    $\tjudge\Gamma{G_1}{C_i}{G^\circ_i}$ and
    $\tjudge\Gamma{G^\circ_1\cap G^\circ_2}{C^\circ}{G'_1}$.

    Suppose $e[\sigma(\pid p)/\pcont]=\m{true}$ (the other case is similar).
    Then $G_1,\gencondE;C^\circ,\sigma\tod G_1,C_1\fatsemi C,\sigma$.
    Since $G^\circ_1\cap G^\circ_2\subseteq G^\circ_1$, Lemma~\ref{lem:seq} allows us to
    conclude that $\tjudge\Gamma{G_1}{C_1\fatsemi C}{G'_1\cup G^\circ_1}$, whence the
    thesis follows by taking $\Gamma'=\Gamma$, $C'=C_1\fatsemi C$, $G_2=G_1$ and
    $G'_2=G'_1\cup G^\circ_1$.
  \end{itemize}
\end{proof}

\begin{proof}[Theorem~\ref{thm:type}]
  If $C\precongrd\nil$, then the first case holds.
  Assume that $C\not\precongrd\nil$; we show that the second case holds by induction on
  the proof of $\tjudge\Gamma{G_1}C{G'_1}$.
  By hypothesis, the last rule applied in this proof cannot be \rname{T}{End}; the cases
  where the last rule applied is \rname{T}{Start}, \rname{T}{Com}, \rname{T}{Sel},
  \rname{T}{Tell} or \rname{T}{Cond} follow immediately from Lemma~\ref{lem:type}, while
  the case of rule \rname{T}{EndSeq} is straightforward from the induction hypothesis.

  We focus on the case of rule \rname{T}{Call}.
  In this case, $C$ has the form $\gencallP;C^\circ$, and we know that
  $\Gamma\vdash X(\wtil{\pid q^T}):(G_X\cseq G'_X)$, $\Gamma\vdash\pids p:T$,
  $G_X[\pids p/\pids q]\subseteq G_1$ and
  $\tjudge\Gamma{G_1\cup(G'_X[\pids p/\pids q])}{C^\circ}{G'_1}$.
  From the hypothesis that $\Gamma\vdash\defs$ we also know that
  $\tjudge{\Gamma_X}{G_X}{C_X}{G'_X}$, where $C_X$ is the body of $X$ as defined in
  $\defs$.
  By Lemma~\ref{lem:sub},
  $\tjudge{\Gamma_X[\pids p/\pids q]}{G_X[\pids p/\pids q]}{C_X[\pids p/\pids q]}{G'_X[\pids p/\pids q]}$,
  whence by Lemma~\ref{lem:mon} also
  $\tjudge{\Gamma}{G_1}{C_X[\pids p/\pids q]}{G'_X[\pids p/\pids q]\cup G_1}$.
  By applying rule \rname{C}{Unfold}, we conclude that
  $\gencallP;C^\circ\precongrd C_X[\pids p/\pids q]\fatsemi C^\circ$, and
  Lemma~\ref{lem:seq} allows us to conclude that
  $\tjudge\Gamma{G_1}{C_X[\pids p/\pids q]\fatsemi C^\circ}{G'_1}$.
  Since procedure calls are guarded, $C_X$ does not begin with a procedure call, and Lemma~\ref{lem:type}
  establishes the thesis.
\end{proof}

\begin{proof}[Theorem~\ref{thm:dec}]
  The proof of this result proceeds in several stages.
  We first observe that deciding whether $\tjudge\Gamma GC{G'}$ is completely mechanical,
  as the typing rules are deterministic.
  Furthermore, those rules can also be used to construct $G'$ from $G$ and $C$; therefore,
  the key step of this proof is showing, given $\Gamma$ and $\langle\defs,C\rangle$, how
  to find a ``canonical typing'' for the recursive definitions, the set $\Gamma_\defs$,
  such that $\Gamma_\defs\vdash\defs$ and $\tjudge{\Gamma,\Gamma_\defs}{G_C}C{G'}$ (with
  $G'$ inferred) iff $\tjudge{\Gamma,\Gamma'}{G_C}C{G''}$ for some $\Gamma'$ and $G''$.
  More precisely, we need to find graphs $G_X$ and $G'_X$ for each procedure $X$ defined
  in $\defs$.

  Our proof proceeds in three steps.
  First, for each $X$ we compute an underapproximation $G^\circ_X$ of the output graph
  $G'_X$, containing all the relevant connections that executing $X$ can add.
  Using this, we are able to compute the input graph $G_X$ and the output graph
  $G'_X=G_X\cup G^\circ_X$.
  Both these steps are achieved by computing a minimal fixpoint of a monotonic operator in
  the set of all graphs whose vertices are the parameters of $X$.
  Finally, we argue that the typing $X:G_X\cseq G_X$ is minimal, and therefore the set
  $\Gamma_\defs$ of all such typings fulfills the property we require.

  Throughout the remainder of this proof, we assume
  $\defs=\{X_i(\wtil{\pid q^i})=C_i\mid i=1,\ldots,n\}$.
  \begin{enumerate}
  \item In order to compute $G^\circ_{X_i}$, we define an auxiliary function
    \lcfaname\ with intended meaning as follows: \lcfa{C_j}G\ computes
    the communication graph obtained from $G$ after one execution of the body of $X_j$,
    assuming that $X_i(\wtil{\pid q_i}):\emptyset\cseq G_i$ for all $i$ and ignoring newly
    created processes.
    We use a conditional union operator $\uplus$ where $G\uplus\{e\}$ denotes $G\cup\{e\}$
    if $e$ is an edge connecting two vertices in $G$, and $G$ otherwise.
    The function $\lcfaname$ is defined as follows.
    {\small
    \begin{align*}
      \lcfa\nil G &= G \\
      \lcfa{\nil;C}G &= \lcfa CG \\
      \lcfa{\gencomf;C}G &= \lcfa CG
    \end{align*}
    \begin{align*}
      \lcfa{\gensel;C}G &= \lcfa CG \\
      \lcfa{\genstart;C}G &= \lcfa CG \\
      \lcfa{\gentellsubscript;C}G &= \lcfa C{G\uplus\{\pid q\leftrightarrow \pid r\}} \\
      \lcfa{\gencondE;C}G &= \lcfa C{\lcfa{C_1}G\cap\lcfa{C_2}G} \\
      \lcfa{\callP{X_i}{\pids p};C}G &= \lcfa C{G\uplus G_i[\pids p/\wtil{\pid q^i}]}
    \end{align*}}

    Using $\lcfaname$, we define an operator $\mathcal T_\lcfaname$ over the set
    $\mathcal G$ of tuples of graphs over the parameters of $X_i$, i.e.\ 
    $\mathcal G=\{\wtil{G_i}\mid G_i\mbox{ is a graph over $\wtil{\pid q^i}$}\}$.
    Observe that $\mathcal G$ is a complete lattice wrt componentwise inclusion.
    \[\mathcal T_\lcfaname(\wtil{G_i}) = \wtil{\lcfa{C_i}{G_i}}\]
    This operator is monotonic, since \lcfaname\ only adds edges to its argument,
    and thus has a least fixpoint that can be computed by iterating $\mathcal T_\lcfaname$
    from the tuple of empty graphs over the right sets of vertices.
    Furthermore, since $\mathcal G$ is finite (each graph has a finite number of vertices)
    this fixpoint corresponds to a finite iterate, and can thus be computed in finite
    time.
    We denote this fixpoint by $\wtil{G^\circ_{X_i}}$.

  \item The construction of the input graphs $G_{X_i}$ follows the same idea: we go
    through the $C_i$s noting the edges that are required for all communications to be
    able to take place.
    It is however slightly more complicated, because we have to keep track of edges that
    the choreography adds to the graph; we therefore need a function \lcfbname\ that
    manipulates two graphs instead of one.
    More precisely, $\lcfb CGG$ returns the graph extending $G$ that is needed for
    correctly executing $C$ (ignoring newly created processes); the first argument keeps
    track of the edges that need to be added to $G$, and the second argument keeps track
    of edges added by executing $C$.
    This function uses the graphs $G^\circ_{X_i}$ computed earlier, which explains why it
    has to be defined afterwards.
    We use the same notational conventions as above, and let $\fst\langle a,b\rangle=a$
    and $\snd\langle a,b\rangle=b$.
    \begin{figure*}
    \scriptsize
    \begin{align*}
      \lcfb\nil{G_a}{G_b} &= \langle{G_a},{G_b}\rangle \\
      \lcfb{\nil;C}{G_a}{G_b} &= \lcfb C{G_a}{G_b} \\
      \lcfb{\gencomf;C}{G_a}{G_b} &= %
        \begin{cases}
          \lcfb C{G_a}{G_b} & \mbox{if $\pid p\leftrightarrow\pid q\in G_b$} \\
          \lcfb C{G_a\uplus\{\pid p\leftrightarrow\pid q\}}{G_b\uplus\{\pid p\leftrightarrow\pid q\}} & \mbox{otherwise}
        \end{cases} \\
      \lcfb{\gensel;C}{G_a}{G_b} &= %
        \begin{cases}
          \lcfb C{G_a}{G_b} & \mbox{if $\pid p\leftrightarrow\pid q\in G_b$} \\
          \lcfb C{G_a\uplus\{\pid p\leftrightarrow\pid q\}}{G_b\uplus\{\pid p\leftrightarrow\pid q\}} & \mbox{otherwise}
        \end{cases} \\
      \lcfb{\genstart;C}{G_a}{G_b} &= \lcfb C{G_a}{G_b} \\
      \lcfb{\gentellsubscript;C}{G_a}{G_b} &= %
        \begin{cases}
          \lcfb C{G_a}{G_b\uplus\{\pid q\leftrightarrow\pid r\}}
          & \mbox{if $\pid p\leftrightarrow\pid  q,\pid p\leftrightarrow\pid r\in G_b$} \\
          \lcfb C{G_a\uplus\{\pid p\leftrightarrow\pid q\}}{G_b\uplus\{\pid p\leftrightarrow\pid q,\pid q\leftrightarrow\pid r\}}
          & \mbox{if $\pid p\leftrightarrow\pid q\not\in G_b$, $\pid p\leftrightarrow\pid r\in G_b$} \\
          \lcfb C{G_a\uplus\{\pid p\leftrightarrow\pid r\}}{G_b\uplus\{\pid p\leftrightarrow\pid r,\pid q\leftrightarrow\pid r\}}
          & \mbox{if $\pid p\leftrightarrow\pid q\in G_b$, $\pid p\leftrightarrow\pid r\not\in G_b$} \\
          \lcfb C{G_a\uplus\{\pid p\leftrightarrow\pid q,\pid p\leftrightarrow\pid r\}}{G_b\uplus\{\pid p\leftrightarrow\pid q,\pid p\leftrightarrow\pid r,\pid q\leftrightarrow\pid r\}}
          & \mbox{if $\pid p\leftrightarrow\pid q,\pid p\leftrightarrow\pid r\not\in G_b$} \\
        \end{cases} \\
      \lcfb{\gencondE;C}{G_a}{G_b} &=
        \lcfb C%
              {\fst(\lcfb{C_1}{G_a}{G_b})\cup\fst(\lcfb{C_2}{G_a}{G_b})}%
              {\snd(\lcfb{C_1}{G_a}{G_b})\cap\snd(\lcfb{C_2}{G_a}{G_b})} \\
      \lcfb{\callP{X_i}{\pids p};C}{G_a}{G_b} &=
        \lcfb C%
              {G_a\uplus(G_i[\pids p/\wtil{\pid q^i}]\setminus G_b)}%
              {G_b\uplus G_i[\pids p/\wtil{\pid q^i}]\uplus G^\circ_i[\pids p/\wtil{\pid q^i}]}
    \end{align*}
    \caption{Definition of \lcfbname\ (case 2 in the proof of Theorem~\ref{thm:dec}).}
    \label{fig:bck}
    \end{figure*}
    The definition of \lcfbname\ is given in Figure~\ref{fig:bck}; it is not the simplest possible, but the formulation
    given is sufficient for our purposes.

    Again we define a monotonic operator over the same $\mathcal G$ as above.
    \[\mathcal T_\lcfbname(\wtil{G_i}) = \wtil{\fst(\lcfb{C_i}{G_i}{G_i})}\]
    We do not need to recompute $G^\circ_i$, since these graphs contain all
    edges that can possibly be added by executing $C_i$.
    The least fixpoint of $\mathcal T_\lcfbname$ can again be computed by finitely
    iterating this operator, and it is precisely $\wtil{G_{X_i}}$.
    We then define $G'_{X_i}=G_{X_i}\cup G^\circ_{X_i}$.

  \item We now show that $\Gamma_\defs=\{X_i(\wtil{\pid q_i}):G_{X_i}\cseq G'_{X_i}\}$ is
    a minimal typing of $\defs$, in the sense explained earlier.
    Observe that it is possible that $\Gamma_\defs\not\vdash\defs$, in particular if the
    $X_i$ are ill-formed choreographies.

    Suppose that $\tjudge{\Gamma,\Gamma'}{G_C}C{G}$ for some $\Gamma'$ and $G$.
    We argue that $\tjudge{\Gamma,\Gamma_\defs}{G_C}C{G'}$, where $G'$ is inferred from
    the typing rules.
    For each procedure $X_i(\wtil{\pid q_i})=C_i$, there must be a unique typing
    $X_i(\wtil{\pid q_i}):G^\ast_{X_i}\cseq G^{\ast\ast}_{X_i}$ in $\Gamma'$.
    By a simple inductive argument one can show that $G^{\circ}_{X_i}\subseteq G^{\ast\ast}$
    (since $\emptyset\subseteq G^{\ast\ast}_{X_i}$ and $\mathcal T_\lcfaname$ preserves
    inclusion in $G^{\ast\ast}_{X_i}$).
    Similarly, one shows that $G_{X_i}\subseteq G^\ast_{X_i}$ and that
    $G^{\ast\ast}_{X_i}\setminus G^\ast_{X_i}\subseteq G'_{X_i}\setminus G_{X_i}$.
    As a consequence, the typing derivation for $\tjudge{\Gamma,\Gamma'}{G_C}C{G}$ can be
    used for $\tjudge{\Gamma,\Gamma_\defs}{G_C}C{G'}$, as all applications of rule
    \rname{T}{Call} are guaranteed to be valid (their preconditions hold) and to produce
    the same results (they change the communication graph in the same way).
  \end{enumerate}
\end{proof}

A consequence of this result is that we can obtain a type inference algorithm
(Theorem~\ref{thm:inference}), as we only need to ``guess'' types for the processes in the
choreography.
As a corollary of the proof, we can also infer the types for parameters of procedural
definitions and freshly created processes.

\begin{proof}[Theorem~\ref{thm:inference}]
  Construct $\Gamma$ by going through $C$ and adding $\pid p:T_{\pid p}$ every time there
  is an action that depends on $\pid p$'s type (i.e.~$\pid p$ is a sender or receiver in a
  communication, or an argument of a procedure call).
  If $\Gamma$ contains two different types for any process, then output \verb+NO+, else
  output $\Gamma$.
  This algorithm will not necessarily assign a type to all processes in $C$, in case $C$
  contains processes whose memory is never accessed.
\end{proof}

\begin{proof}[Theorem~\ref{cor:inference}]
  Inferring the types of freshly created processes is analogous to the previous proof.

  As for parameters of procedure definitions, we omit the details of the proof, as it
  repeats ideas previously used.
  Define an operator $\mathcal T_\types$ over tuples of typing contexts (one for each
  $X_i$ defined in $\defs$) that generates a typing context for each $X_i$ in the same way
  as in the previous proof.
  If any contradictions are found, then fail.
  Iterate $\mathcal T_\types$ until either failure occurs (in which case the $X_i$s are
  not properly defined) or a fixpoint is reached.
  Finally, assign a random type (e.g.~$\mathbb N$) to each process variable that has not
  received a type during this procedure.
  The algorithm readily extends to infer the types of processes created inside procedure
  definitions.
\end{proof}


\subsection{EndPoint Projection}

\paragraph{Merging.}
The full definition of merging is given in Figure~\ref{fig:merge}.
\begin{figure*}
\centering\small
\begin{align*}
&
\left(\asend{\pid q}{e};B\right) \ \merge \ \left(\asend{\pid q}{e};B'\right)
 \ = \ \asend{\pid q}{e};\left(B\merge B'\right)
\qquad
\left(\arecv{\pid p}f;B\right) \ \merge \ \left(\arecv{\pid p}f;B'\right)
 \ = \ \arecv{\pid p}f;\left(B\merge B'\right)
\\[1ex]
&\left(\atells{\pid q}{\pid r};B\right) \ \merge \ \left(\atells{\pid q}{\pid r};B'\right)
 \ = \ \atells{\pid q}{\pid r};\left(B\merge B'\right)
\qquad
\left(\atellr{\pid q}{\pid r};B\right) \ \merge \ \left(\atellr{\pid q}{\pid r};B'\right)
 \ = \ \atellr{\pid q}{\pid r};\left(B\merge B'\right)
\\[1ex]
&\left(\asel{\pid q}{l};B\right)\ \merge \ \left(\asel{\pid q}{l};B'\right)
\ = \ \asel{\pid q}{l};\left(B\merge B'\right)
\qquad
\left(\gencallP;B\right)\ \merge \ \left(\gencallP;B'\right) \ = \ \gencallP;\left(B\merge B'\right)
\\[1ex]
&\left(\astart{\pid q}{B_2};B_1\right) \ \merge \ \left(\astart{\pid q}{B'_2};B'_1\right)
 \ = \ \astart{\pid q}{\left(B_2\merge B'_2\right)};\left(B_1 \merge B'_1\right)
\\[1ex]
&B_1 \ \merge\ B_2 \ = \ B'_1 \ \merge\ B'_2 \quad \left(\text{if } B_1 \precongr B'_1 \text{ and } B_2 \precongr B'_2\right)
\qquad
\left(\nil;B\right) \ \merge \ B'
 \ = \ B\merge B'
\\[1ex]
&\left(\cond{e}{B_1}{B_2}\right);B \ \merge \
\left(\cond{e}{B'_1}{B'_2}\right);B' \ = \
\left(\cond{e}{(B_1 \merge B'_1)}{(B_2 \merge B'_2)}\right);\left(B\merge B'\right)
\\[1ex]
&
  \left(\abranch{\pid p}{\{l_i:B_i\}_{i\in J}};B\right) \merge
  \left(\abranch{\pid p}{\{l_i:B'_i\}_{i\in K}};B'\right) =
  \abranch{\pid p}{\left(\{l_i:(B_i \merge B'_i)\}_{i\in J\cap K}\cup\{l_i:B_i\}_{i\in J \setminus K}\cup\{l_i:B'_i\}_{i\in K \setminus J}\right)};\left(B\merge B'\right)
\end{align*}
\caption{Merging operator in PP.}
\label{fig:merge}
\end{figure*}

\paragraph{Projections of procedure definitions.}
When projecting procedure definitions, we simply kept all arguments, and relied on typing to guarantee that no process
will ever attempt to communicate with another process it does not know.
However, it is ``cleaner'' to refine the definition of projection so that such parameters are not even formally used.

This is technically not challenging, as we can use typing information to decide which parameters should be kept in each
projection.
In other words, when computing $\epp\gendef{}$, the projected procedure $X_i$ will only contain as arguments those
$\pid q_j$ such that $\knows[G_X]{\pid q_i}{\pid q_j}$, where $G_X$ is given by typing.
(This definition is non-deterministic, but it can be made deterministic by using the minimal graph $G_X$ computed by
the type inference algorithm.)
A simple annotation can then ensure that projected procedure calls only keep the arguments in the corresponding
positions, and typing guarantees that all those arguments are known at runtime by the process invoking the recursive
call.

An implementation of Quicksort showcasing this reasoning was previously published in~\cite{CM16a}.

\paragraph{EPP.} We sketch the proof of the EPP Theorem.

\begin{proof}[Theorem~\ref{thm:epp} (sketch)]
The structure of the proof is standard, from~\cite{M13:phd}, so we only show the most interesting differing details.
In particular, we need to be careful about how we deal with connections, which is a new key ingredient in PC wrt to
previous work.
We demonstrate this point for the direction of \emph{(Completeness)}; the direction for \emph{(Soundness)} is proven
similarly.
The proof proceeds by induction on the derivation of $G, C,\sigma \tod G'', C',\sigma'$. The interesting cases are
reported below.
\begin{itemize}
\item \rname{C}{Tell}:
From the definition of EPP we get:
\begin{multline*}
\epp{\gentell;C^\circ,\sigma}{} \precongr
\\
\actor{\pid p}{\sigma(\pid p)}{
	\atells{\pid q}{\pid r}; \epp{C^\circ}{\pid p}
}
\ \parp\
\actor{\pid q}{\sigma(\pid q)}{
	\atellr{\pid p}{\pid r}; \epp{C^\circ}{\pid q}
}
\ \parp\
\actor{\pid r}{\sigma(\pid r)}{
	\atellr{\pid p}{\pid q}; \epp{C^\circ}{\pid r}
}
\ \parp\
N
\end{multline*}
By \rname{P}{Tell} we get:
\begin{multline*}
\epp{\gentell;C^\circ,\sigma}{} \to \\
\actor{\pid p}{\sigma'(\pid p)}{
	\epp{C^\circ}{\pid p}
}
\ \parp\
\actor{\pid q}{\sigma'(\pid q)}{
	\epp{C^\circ}{\pid q}
}
\ \parp\
\actor{\pid r}{\sigma'(\pid r)}{
	\epp{C^\circ}{\pid r}
}
\ \parp\
N
\end{multline*}
which proves the thesis, since we can assume that the projection of $C'$ remains unchanged for the other processes ($N$
stays the same).

\item \rname{C}{Start}:
This is the most interesting case.
From the definition of EPP we get:
\[
\epp{\genstart;C^\circ,\sigma}{} \precongr
\actor{\pid p}{\sigma(\pid p)}{
	\astart{\pid q^T}{\epp{C^\circ}{\pid q}}; \epp{C^\circ}{\pid p}
}
\ \parp\
N
\]
From the semantics of PP we get:
\[
\epp{\genstart;C^\circ,\sigma}{} \to
\actor{\pid p}{\sigma'(\pid p)}{
	(\epp{C^\circ}{\pid p})[\pid q'/\pid q]
}
\ \parp\
\actor{\pid q'}{\bot_T}{
	\epp{C^\circ}{\pid q'}
}
\ \parp\
N
\]
The Barendregt convention implies $C \precongr (\genstart;C^\circ)[\pid q'/\pid q]$.
We now have to prove that:
\[
\epp{C^\circ[\pid q'/\pid q],\sigma}{} \precongr
\actor{\pid p}{\sigma'(\pid p)}{
	(\epp{C^\circ}{\pid p})[\pid q'/\pid q]
}
\ \parp\
\actor{\pid q'}{\bot_T}{
	\epp{C^\circ}{\pid q'}
}
\ \parp\
N
\]
We observe that this is true only if process $\pid q$ does not occur free in $N$, i.e., $\pid q$ appear in $N$ only
inside the scope of a binder. The latter must be of the form $\atellr{\pid r}{\pid q};B$.
This is guaranteed by the fact that $C$ is well-typed, since the typing rules prevent other processes in $N$ to
communicate with $\pid q$ without being first introduced.\qedhere
\end{itemize}
\end{proof}

\paragraph{Choreography Amendment.}
We now define precisely an amendment function that makes every choreography projectable.
As mentioned earlier, we follow ideas previously used in other choreography models~\cite{LMZ13,ourstuff}.
Observe that, by definition, the only choreography construct that can lead to unprojectability is the conditional.
\begin{definition}[Amendment]
Given a choreography $C$, the transformation $\amend(C)$ repeatedly applies the following procedure until it is no
longer possible, starting from the inner-most subterms in $C$.
For each conditional subterm $\gencondE$ in $C$, let $\pids r \subseteq (\pn(C_1)\cup\pn(C_2))$ be such that
$\epp{C_1}{\pid r} \merge \epp{C_2}{\pid r}$ is undefined for all $\pid r \in \pids r$; then $\gencond$ in $C$ is
replaced with:
\[
 \begin{array}{ll}
 \m{if}\, \pid p.e \, & \m{then} \, \sel{\pid p}{\pid r_1}{\lleft};\cdots;\sel{\pid p}{\pid r_n}{\lleft};C_1
\\
 & \m{else} \, \,\,
 {\sel{\pid p}{\pid r_1}{\lright};\cdots;\sel{\pid p}{\pid r_n}{\lright}};C_2
 \end{array}
\]
\end{definition}

By the definitions of $\amend$ and EPP and the semantics of PC, we get the following properties, where $\to^*$ is the
transitive closure of $\to$.
\begin{theorem}[Amendment]
  \label{lem:amend}
Let $C$ be a choreography. Then:
\begin{description}
\item[(Completeness)] $\amend(C)$ is defined;
\item[(Projectability)] for all $\sigma$, $\epp{\amend(C),\sigma}{}$ is defined;
\item[(Correspondence)] for all $G$, $\sigma$ and $\defs$:
  \begin{itemize}
  \item if $G,C,\sigma\tod G',C',\sigma$, then $G,\amend(C),\sigma\tod^*G',\amend(C'),\sigma'$;
  \item if $G,\amend(C),\sigma\tod G',C',\sigma'$, then there exist $C''$ and $\sigma''$ such that
    $G,C,\sigma\tod G',C'',\sigma''$ and $G',C',\sigma'\tod^\ast G',C'',\sigma''$.
\end{itemize}
\end{description}
\end{theorem}
\begin{proof}[Theorem~\ref{lem:amend}]
  (Completeness) and (Projectability) are immediate by definition of $\amend$.
  For (Correspondence), the proof is by analysis of the possible transitions.
  The only interesting cases occur when the transition consumes a conditional.
  In the case $G,C,\sigma\to G',C',\sigma$, then $\amend(C)$ also has to consume the label selections introduced by
  amendment in the branch taken in order to reach $\amend(C')$.
  Conversely, if $G,\amend(C),\sigma$ makes a transition that consumes a conditional, then $C'$ needs to consume the
  label selections introduced by amendment in order to match the corresponding move by $C$.
\end{proof}

This result also holds if we replace $\defs$ by $\amend(\defs)$ in the relevant places.

The first choreography in Remark~\ref{remark:selections}, which is unprojectable, can be amended to the projectable
choreography presented at the end of the same remark.

\subsection{Asynchrony}
We detail the whole sets of rules for the semantics of aPC (Figure~\ref{fig:apc_semantics}) and its type system (Figure~\ref{fig:apc_types}), as well as for the semantics of aPP (Figure~\ref{fig:app_semantics}).
The full definition of EPP is given in Figures~\ref{fig:aepp} and~\ref{fig:aepp_state}.
\begin{figure}[t]
\footnotesize
\begin{eqnarray*}
&\infer[\rname{C}{Com-S}]
{
  G,\genasend;C,\sigma
  \ \tod \ 
  G, C[v/x], \sigma
}
{
  \know{\pid p}{\pid q} &
  \eval{e[\sigma(\pid p)/\pcont]}v
}
\\[1ex]
&\infer[\rname{C}{Com-R}]
{
  G,\asrecv{\pid p}{\pid q}fv;C,\sigma
  \ \tod \ 
  G, C, \sigma[\pid q \mapsto w]
}
{
\know{\pid q}{\pid p}
&
  \eval{f[\sigma(\pid q)/\pcont](v)}w
}
\\[1ex]
&\infer[\rname{C}{Sel-S}]
{
  G,\genasels;C,\sigma \ \tod \ G,C[l/x],\sigma
}
{
  \know{\pid p}{\pid q}
}
\\[1ex]
&\infer[\rname{C}{Sel-R}]
{
  G,\acom{\bullet_{\pid p}}{\pid q[l]}l;C,\sigma \ \tod \ G,C,\sigma
}
{
\know{\pid q}{\pid p}
}
\\[1ex]
&\infer[\rname{C}{Tell-S}]{
  G, \genatell;C, \sigma \ \tod\ G, C[\pid q/x, \pid r/y], \sigma
}{
  \know{\pid p}{\pid q}
  &
  \know{\pid p}{\pid r}
}
\\[1ex]
&\infer[\rname{C}{Tell-R}]{
  G, \acom{\bullet_{\pid p}.\pid r}{\pid q}{\pid r};C, \sigma
  \ \tod\ 
  \updates{\pid q}{\pid r}, C, \sigma
}{
\know{\pid q}{\pid p}
}
\end{eqnarray*}
\caption{Asynchronous PC, Semantics of New Runtime Terms.}
\label{fig:apc_semantics}
\end{figure}

\begin{figure}[t]
\footnotesize
\begin{eqnarray*}
&\infer[\rname{T}{Com-S}]
{
  \Gamma;G\seq\genasend;C \cseq G'
}
{
  \know{\pid p}{\pid q} &
  \Gamma\seq\pid p:T_{\pid p} &
  \pcont:T_{\pid p}\vdash_\types e:T &
  \Gamma\oplus(x:T);G\seq C \cseq G'
}
\\[1ex]
&\infer[\rname{T}{Com-RV}]
{
  \Gamma;G\seq\asrecv{\pid p}{\pid q}fx;C \cseq G'
}
{
  \know{\pid q}{\pid p} &
  \Gamma\seq\pid q:T_{\pid q} &
  \pcont:T_{\pid q}\vdash_\types f:T\to T_{\pid q} &
  \Gamma\oplus(x:T);G\seq C \cseq G'
}
\\[1ex]
&\infer[\rname{T}{Com-RT}]
{
  \Gamma;G\seq\asrecv{\pid p}{\pid q}fv;C \cseq G'
}
{
  \know{\pid q}{\pid p} &
  \Gamma\seq\pid q:T_{\pid q} &
  \vdash_\types v:T &
  \pcont:T_{\pid q}\vdash_\types f:T\to T_{\pid q} &
  \Gamma;G\seq C \cseq G'
}
\\[1ex]
&\infer[\rname{T}{Sel-S}]
{
  \Gamma;G\vdash\genasels;C \cseq G'
}
{
  \know{\pid p}{\pid q} &
  \Gamma\oplus(x=l);G\vdash C \cseq G'
}
\\[1ex]
&\infer[\rname{T}{Sel-RV}]
{
  \Gamma;G\vdash\acom{\bullet_{\pid p}}{\pid q[l]}x;C \cseq G'
}
{
  \know{\pid q}{\pid p} &
  \Gamma\oplus(x=l);G\vdash C \cseq G'
}
\quad
\infer[\rname{T}{Sel-RT}]
{
  \Gamma;G\vdash\acom{\bullet_{\pid p}}{\pid q[l]}l;C \cseq G'
}
{
  \know{\pid q}{\pid p} &
  \Gamma;G\vdash C \cseq G'
}
\\[1ex]
&\infer[\rname{T}{Tell-S}]{
  \Gamma;G\vdash\genatell;C \cseq G'
}{
  \know{\pid p}{\pid q} &
  \know{\pid p}{\pid r} &
  \Gamma\oplus(x=\pid q)\oplus(y=\pid r) \vdash C \cseq G'
}
\\[1ex]
&\infer[\rname{T}{Tell-RV}]{
  \Gamma;G\vdash\acom{\bullet_{\pid p}.\pid r}{\pid q}{\pid r};C \cseq G'
}{
  \know{\pid q}{\pid p} &
  \Gamma;\updates{\pid q}{\pid r} \vdash C \cseq G'
}
\\[1ex]
&\infer[\rname{T}{Tell-RT}]{
  \Gamma;G\vdash\acom{\bullet_{\pid p}.\pid r}{\pid q}{x};C \cseq G'
}{
\know{\pid q}{\pid p} &
  \Gamma\oplus(x=\pid r);\updates{\pid q}{\pid r} \vdash C \cseq G'
}
\end{eqnarray*}
\caption{Asynchronous PC, Typing Rules (New Runtime Terms).}
\label{fig:apc_types}
\end{figure}

\begin{figure*}
\footnotesize
\begin{eqnarray*}
&\infer[\rname{P}{Com-S}]
{
  \actor[\rhop]{\pid p}{v}{\asend{\pid q}{e};B_{\pid p}} \parp \actor[\rhoq]{\pid q}w{B_{\pid q}}
  \ \tob \ 
  \actor[\rhop]{\pid p}{v}{B_{\pid p}} \parp \actor[\rho'_{\pid q}]{\pid q}w{B_{\pid q}}
}
{
  \rho'_{\pid q}=\rhoq\cdot\langle\pid p,e{[v/\pcont]}\rangle
}
\\[1ex]
&\infer[\rname{P}{Com-R}]
{
  \actor[\rhoq]{\pid q}{w}{\arecv{\pid p}{f};B}
  \ \tob \ 
  \actor[\rho'_{\pid q}]{\pid q}{u}{B}
}
{
  \rhoq\preceq\langle\pid p,v\rangle\cdot\rho'_{\pid q} &
  u = (f[w/\pcont])(v)
}
\\[1ex]
&\infer[\rname{P}{Tell-S}]
{
  \actor[\rhop]{\pid p}{v}{\atells{\pid q}{\pid r};B_{\pid p}} \parp \actor[\rhoq]{\pid q}w{B_{\pid q}} \parp \actor[\rho_{\pid r}]{\pid r}z{B_{\pid r}}
  \ \tob \ 
  \actor[\rhop]{\pid p}{v}{B_{\pid p}} \parp \actor[\rho'_{\pid q}]{\pid q}w{B_{\pid q}} \parp \actor[\rho'_{\pid r}]{\pid r}z{B_{\pid r}}
}
{
  \rho'_{\pid q}=\rhoq\cdot\langle\pid p,\pid r\rangle &
  \rho'_{\pid r}=\rho_{\pid r}\cdot\langle\pid p,\pid q\rangle
}
\\[1ex]
&\infer[\rname{P}{Tell-R}]
{
  \actor[\rhoq]{\pid q}{w}{\atellr{\pid p}{\pid r};B}
  \ \tob \ 
  \actor[\rho'_{\pid q}]{\pid q}{w}{B}
}
{
  \rhoq\preceq\langle\pid p,\pid r\rangle\cdot\rho'_{\pid q} &
}
\\[1ex]
&\infer[\rname{P}{Sel-S}]
{
  \actor[\rhop]{\pid p}{v}{\asel{\pid q}{l};B} \parp \actor[\rhoq]{\pid q}w{B_{\pid q}}
  \ \tob \
  \actor[\rhop]{\pid p}{v}{B} \parp \actor[\rho'_{\pid q}]{\pid q}w{B_{\pid q}}
}
{
  \rho'_{\pid q}=\rhoq\cdot\langle\pid p,l\rangle
}
\\[1ex]
&\infer[\rname{P}{Sel-R}]
{
  \actor[\rhoq]{\pid q}{w}{\abranch{\pid p}{\{ l_i : B_i\}_{i\in I}}}
  \ \tob \
  \actor[\rho'_{\pid q}]{\pid q}{w}{B_j}
}
{
  \rhoq\preceq\langle\pid p,l_j\rangle\cdot\rho'_{\pid q} &
  j \in I
}
\end{eqnarray*}
\caption{Asynchronous Procedural Processes, Semantics (New Rules).}
\label{fig:app_semantics}
\end{figure*}

\begin{figure}[t]
\footnotesize
\begin{eqnarray*}
&\epp{\genasend;C}{\pid r}=
  \begin{cases}
    \asend{\pid q}{e};\epp{C}{\pid r} & \text{if } \pid r = \pid p \\
    \epp{C}{\pid r} & \text{otherwise}
  \end{cases}
\qquad
\epp{\genarecv;C}{\pid r}=
  \begin{cases}
    \arecv{\pid p}{f};\epp{C}{\pid r} & \text{if } \pid r = \pid q \\
    \epp{C}{\pid r} & \text{otherwise}
  \end{cases}
\\[1ex]
&\epp{\genasels;C}{\pid r}=
\begin{cases}
  \asel{\pid q}{l} & \text{if } \pid r = \pid p \\
  \epp{C}{\pid r} & \text{otherwise}
\end{cases}
\qquad
\epp{\genaselr;C}{\pid r}=
\begin{cases}
  \abranch{\pid p}{\{ l : \epp{C}{\pid r} \}} & \text{if } \pid r = \pid q \\
  \epp{C}{\pid r} & \text{otherwise}
\end{cases}
\\[1ex]
&\epp{\genatell;C}{\pid s}=
\begin{cases}
  \atells{\pid q}{\pid r};\epp{C}{\pid s} & \text{if } \pid s = \pid p
  \\
  \epp{C}{\pid s} & \text{otherwise}
\end{cases}
\qquad
\epp{\genatold;C}{\pid s}= 
\begin{cases}
  \atellr{\pid p}{\pid r};\epp{C}{\pid s} & \text{if } \pid s = \pid q
  \\
  \epp{C}{\pid s} & \text{otherwise}
\end{cases}
\end{eqnarray*}
\caption{Asynchronous PC, Behaviour Projection (New Rules).}
\label{fig:aepp}
\end{figure}

\begin{figure}[t]
\footnotesize
\begin{eqnarray*}
&\something{\asrecv{\pid p}{\pid q}fv;C}{\pid r}=
  \begin{cases}
    \langle\pid p,v\rangle\cdot\something{C}{\pid r} & \text{if } \pid r = \pid q \\
    \something{C}{\pid r} & \text{otherwise}
  \end{cases}
\qquad
\something{\acom{\bullet_{\pid p}}{\pid q[l]}l;C}{\pid r}=
\begin{cases}
  \langle\pid p,l\rangle\cdot\something{C}{\pid r} & \text{if } \pid r = \pid q \\
  \epp{C}{\pid r} & \text{otherwise}
\end{cases}
\\[1ex]
&\something{\acom{\bullet_{\pid p}.\pid r}{\pid q}{\pid r};C}{\pid s}= 
\begin{cases}
  \langle\pid p,\pid r\rangle\cdot\something{C}{\pid s} & \text{if } \pid s = \pid q
  \\
  \epp{C}{\pid s} & \text{otherwise}
\end{cases}
\\[1ex]
&\something{\gencond;C}{\pid r} = \something{C_1}{\pid r}\cdot\something{C}{\pid r}
\qquad
\something{\eta;C}{\pid r} = \something{I;C}{\pid r} = \something{C}{\pid r}
\end{eqnarray*}
\caption{Asynchronous PC, State Projection.}
\label{fig:aepp_state}
\end{figure}

The proofs of the relationships between PC/PP and their asynchronous counterparts are mechanical.

\begin{proof}[Theorem~\ref{thm:pc-vs-apc}]
  Straightforward by case analysis on the possible transitions of $C$.
\end{proof}

\begin{proof}[Theorem~\ref{thm:pp-vs-app}]
  Straightforward by case analysis on the possible transitions of $N$.
\end{proof}

\subsection{Procedures with holes}
%
We illustrate the use of holes to define a general-purpose iterator, procedure \lstinline+Iter+ below.
Using \lstinline+Iter+, we define a procedure \lstinline+Reverse+ for reversing the list of a process \lstinline+q+.
We omit the straightforward procedures for \lstinline+dec+rementing a counter and \lstinline+pop+ping a list.
\begin{lstlisting}
Iter(p,q,r) = if p.is_zero
  then p -> q,r[stop]
  else dec<p>; p -> q,r[cont]; \Box_h; Iter<p,q,r>

Reverse(q) =
  q starts p,r; q: p<->r;
  q.size -> p; q.empty -> r;
  Iter<p,q,r> with h \mapsto q.top -> r.append; pop<q>;
  r.* -> q
\end{lstlisting}



\end{document}